\documentclass[12pt]{article}

\usepackage{a4wide}
\usepackage{amsmath,amsfonts,amssymb,
amsthm,color, mathabx}
\usepackage{hyperref}
\usepackage{csquotes}

\usepackage[active]{srcltx}

\usepackage{mathrsfs}
\usepackage{amscd,graphicx}

\usepackage{enumitem}
\setlist{leftmargin=*}

\let\OLDthebibliography\thebibliography
\renewcommand\thebibliography[1]{
  \OLDthebibliography{#1}
  \setlength{\parskip}{0pt}
  \setlength{\itemsep}{0pt plus 0.3ex}
}

\newtheorem{theorem}{Theorem}[section]
\newtheorem{definition}[theorem]{Definition}
\newtheorem{proposition}[theorem]{Proposition}

\newtheorem{lemma}[theorem]{Lemma}

\theoremstyle{definition}
\newtheorem{remark}[theorem]{Remark}

\numberwithin{equation}{section}

\def\C{{\mathbb C}}
\def\N{{\mathbb N}}
\def\Z{{\mathbb Z}}
\def\R{{\mathbb R}}
\def\supp{\mathop{\mathrm{supp}} \nolimits}
\newcommand{\tr}{\mathrm{Tr}}

\title{Magnetic Weyl Super Calculus:\\ Schatten-class properties, commutator criterion, and complete positivity\par}

\author{H.~D. Cornean\footnote{Department of Mathematical Sciences, Aalborg University, Thomas Manns Vej 23, 9220 Aalborg, Denmark; e-mail: cornean@math.aau.dk}\, and M.~H. Thorn\footnote{Department of Mathematical Sciences, Aalborg University, Thomas Manns Vej 23, 9220 Aalborg, Denmark; e-mail: mikkelht@math.aau.dk}}

\begin{document}

\maketitle

 \begin{abstract}
    We combine our previous results  on magnetic pseudo-differential operators for Hörmander symbols dominated by tempered weights \cite{Thorn2026} with the magnetic Weyl super calculus of Lee and Lein \cite{LeeLein2022,LeeLein2025}. This allows us to extend some previous results on the semi-super and super Moyal algebra, as well as to prove boundedness, compactness, and Schatten-class properties of super operators. 
    
    Moreover, we prove a Beals-type commutator criterion for super operators and we also formulate sufficient conditions on super symbols in order to give rise to completely positive and trace preserving maps.

    For most of the proofs we use decompositions of operators and super operators based on Parseval frames of smoothing operators.
 \end{abstract}

\tableofcontents

\section{Introduction}

A systematic study of the gauge-covariant magnetic pseudo-differential super calculus was initiated by Lee and Lein \cite{LeeLein2022}, inspired by the magnetic pseudo-differential calculus developed in the 2000s by Iftimie, M{\u a}ntoiu, and Purice \cite{Mantoiu2004,Iftimie2007}.  

It turned out from the seminal work of Feichtinger and Gröchenig \cite{Feichtinger1997,Grochenig2006} that using Gabor frames in the study of ``usual" non-magnetic pseudo-differential operators can shed a completely different light on the theory. Frame decompositions adapted to the magnetic case have also played a crucial role in the study of magnetic pseudo-differential operators \cite{CorneanHelfferPurice2018,CorneanGarde2019,CorneanHelfferPurice2024,Bachmann2025,Thorn2026,Cornean2026}, as well as topics related to exponentially localized magnetic Wannier functions \cite{CorneanMonacoMoscolari2019}. Some of these ideas were later on lifted to the magnetic pseudo-differential super calculus again by Lee and Lein in \cite{LeeLein2025}. 

In our present work we further develop the matrix representation for super operators and obtain new results. A partial motivation for studying this super calculus is to see how it can be extended in order to fit into the theory of  completely positive and trace preserving super operators \cite{Kraus1971,Choi1975} and Lindblad-type super operators \cite{Lindblad1976,GKS1976,Holevo1998,Alazzawi2015}, both of which being indispensable for quantum information theory \cite{Chuang2010}. Studying both through a pseudo-differential lens will hopefully lead to new results on open quantum systems with infinite dimensional Hilbert spaces.

\paragraph{Outline:} In Section \ref{sec:defn} we give a short and compact introduction to the fundamental objects in magnetic pseudo-differential theory, which, for those already well-acquainted with it, should be enough to read the rest of the paper. For those unfamiliar with the subject we also give a more detailed account of certain objects in Subsection \ref{sec:detour}, to be read after the short introduction.

We introduce our main tools in Section \ref{sec:decompositions}, which consists of frame characterizations and decompositions of some operator spaces, along with an infinite matrix representation for super operators given in Theorem \ref{thm:matrix_rep}. These decompositions are used to prove results on: 

\begin{itemize} 
\item The algebra of super symbols in Section \ref{sec:calc}, see Proposition \ref{prop:calc_super_semi} and Proposition \ref{prop:calc_super_moyal}. 
\item Boundedness of super operators in Section \ref{sec:bound}, see Proposition \ref{prop:opspaces}, Proposition \ref{prop:continuoustensor}, and Theorem \ref{thm:caldvail}.
\item Beals-type super commutator criterion in Section \ref{sec:commutator}, see Theorem \ref{thm:beals}.
\item Sufficient conditions on super symbols so that their associated super operators are completely positive and trace preserving in Section \ref{sec:positivity}, see Theorem \ref{thm:cptp}.
\end{itemize}

\section{Preliminaries}\label{sec:defn}

\subsection{The Short Story}
Before going into details, we summarize here the minimal amount of information needed to understand the objects we will work with.  

A regular magnetic field is a smooth closed 2-form $B$ on $\R^d$, $d\in\N$, with components in $BC^\infty(\R^d)$ \cite{Mantoiu2004,CorneanHelfferPurice2024}. We fix such a $B$ and define a magnetic potential, $dA=B$, in the transversal gauge:
\begin{equation*}
    A_k(x)=\sum_{j=1}^d\int_0^1\mathrm{d}s\,s x_j B_{jk}(sx).
\end{equation*}
Since $B$ has components in $BC^\infty(\R^d)$, $A$ has components in $C_{\text{pol}}^\infty(\R^d)$, the space of polynomially bounded smooth functions. We note the important fact that for a fixed $y\in \R^d$, the family of vector potentials 
\begin{equation*}
    A_k(x;y)=\sum_{j=1}^d\int_0^1\mathrm{d}s\,s (x_j-y_j) B_{jk}(y+s(x-y))
\end{equation*}
obeys $dA(\cdot;y)=B$, and it can grow at most linearly in $|x-y|$ together with all its derivatives. 

We denote the flux of $B$ through the triangle $\langle x,y,z\rangle$ with vertices $x,y,z\in\R^d$ by $\Gamma(x,y,z)=\int_{\langle x,y,z\rangle}B$ and the circulation of $A$ along the line segment $[y,x]$ from $y$ to $x$ by $\varphi(x,y)=\int_{[y,x]}A$. We have the important identity:
\begin{equation}\label{eq:magpotential}
    \partial_{j}\phi(\cdot,y)-A_j(\cdot )=-A_j(\cdot; y)
\end{equation}
for all $y\in\R^d$.

We now can define the magnetic Weyl quantization $\mathfrak{op}^A(\phi)$ of a symbol $\phi\in\mathscr{S}'(\R^{2d})$ to be the unique operator with distributional kernel given by \cite{Mantoiu2004}:
\begin{equation*}
    \big(k^A\phi\big)(x,y)=\frac{1}{(2\pi)^d}\int \mathrm{d}\xi\, e^{i\xi\cdot(x-y)}e^{i\varphi(x,y)}\phi\left(\frac{x+y}{2},\xi\right)
\end{equation*}
The map $k^A$, which sends the symbol $\phi$ into the kernel $k^A\phi$ of $\mathfrak{op}^A(\phi)$, is called the magnetic Weyl transform. 

The magnetic super Weyl quantization is then defined by the logic that a tensor product of symbols $\phi\otimes\psi$ should be quantized into the super operator \cite{LeeLein2022}:
\begin{equation*}
    \mathfrak{Op}^A(\phi\otimes\psi)\colon T\mapsto\mathfrak{op}^A(\phi) T\mathfrak{op}^A(\psi)
\end{equation*}
Such a super operator will be denoted by $\mathfrak{op}^A(\phi)\odot\mathfrak{op}^A(\psi)$, indicating that it ``sandwiches" its operator-variable \cite{LeeLein2025}. Tensor products $\phi\otimes\psi$ of symbols in $\mathscr{S}'(\R^{2d})$ are dense in $\mathscr{S}'(\R^{4d})$ allowing us to define a super quantization for any super symbol $\Phi\in\mathscr{S}'(\R^{4d})$. We do this by describing how the super operator $\mathfrak{Op}^A(\Phi)$ changes the kernel of the operator on which it acts. Namely, by  introducing the magnetic super Weyl transform:
\begin{equation*}
    \begin{aligned}
        \big(K^A\Phi\big)(x_L,x_R,y_L,y_R)&=\frac{1}{(2\pi)^{2d}}\int \mathrm{d}\xi_L\,\mathrm{d}\xi_R\, e^{i\left[\xi_L\cdot(x_L-y_L)-\xi_R\cdot(x_R-y_R)\right]}e^{i\left[\varphi(x_L,y_L)-\varphi(x_R,y_R)\right]}\\
        &\quad\times\Phi\left(\frac{x_L+y_L}{2},\xi_L,\frac{x_R+y_R}{2},\xi_R\right),
    \end{aligned}
\end{equation*}
the integral kernel of $\mathfrak{Op}^A(\Phi)T$ should be at least formally given by (here $\mathfrak{int}^{-1}(T)$ denotes the integral kernel of $T$):
$$\big (\mathfrak{int}^{-1}(\mathfrak{Op}^A(\Phi)T)\big )(x_L,x_R)=\int \mathrm{d}y_L\,\mathrm{d}y_R\, \big(K^A\Phi\big)(x_L,x_R,y_L,y_R) \big (\mathfrak{int}^{-1}(T)\big )(y_L,y_R).$$
We will also write 
\begin{equation*}
    \mathfrak{Op}^A(\Phi)=\mathfrak{int}\circ\mathfrak{Int}(K^A\Phi)\circ\mathfrak{int}^{-1},
\end{equation*}
where the map ``$\mathfrak{int}$" takes an integral kernel into an operator, and ``$\mathfrak{Int}$" takes a super integral kernel into a super operator. 

In terms of symbols, we primarily work with the Hörmander classes $s_0(m)$ and super classes $S_0(M)$ of smooth functions dominated by a tempered weight $m$ or $M$ \cite{Nicola2010,Thorn2026}. A tempered weight $m$ is a positive function on $\R^{2d}$, or respectively $\R^{4d}$ for $M$, such that there exists $a,C>0$ for which
\begin{equation*}
    m(X+Y)\leq Cm(X)\langle Y\rangle^a
\end{equation*}
for all $X,Y\in \R^{2d}$.

\subsection{The Detailed Detour}\label{sec:detour}

We introduce almost all definitions and notations used in this article in the following subsections. Note, non-super objects are denoted by small letters, e.g. the Weyl quantization is $\mathfrak{op}$ while the super Weyl quantization is $\mathfrak{Op}$. This is in conflict with notation in most other articles \cite{Mantoiu2004,CorneanHelfferPurice2024,Thorn2026} where the Weyl quantization is denoted by a capital letter $\mathfrak{Op}$, but consisting with the papers on super quantizations \cite{LeeLein2022,LeeLein2025}.

\paragraph{Variable names:} We will use the phase space variables $X,Y,Z\in\R^{2d}$, which will always be split into position and momentum variables as follows $X=(x,\xi)$, $Y=(y,\eta)$, and $Z=(z,\zeta)$. Likewise, we will use variables for super objects $\mathbf{X},\mathbf{Y},\mathbf{Z}\in\R^{4d}$ and decompose them into left and right phase space variables, i.e. $\mathbf{X}=(X_L,X_R)=(x_L,\xi_L,x_R,\xi_R)$. This follows \cite{Iftimie2007,LeeLein2022}.

A similar convention will be used for indices, where $\tilde{\alpha}=(\alpha,\alpha')\in\Z^{2d}$ is an index containing both a position and momentum index. This follows \cite{CorneanHelfferPurice2024,LeeLein2025,Thorn2026} to some degree.

\paragraph{Operator spaces:} For any two topological vector spaces $V,W$, we let $\mathcal{B}(V,W)$ denote the bounded operators between $V$ and $W$. For the spaces considered in this paper, boundedness will always be equivalent to continuity \cite{Osborne2014,Treves2006}. If $V,W$ are Hilbert spaces we let $\mathcal{B}_p(V,W)$ denote the $p$-Schatten-class with $p\in(0,\infty]$ \cite{Conway2000,Treves2006}. Note this means that $\mathcal{B}_\infty(V,W)$ is the space of compact operators.

\paragraph{Magnetic Weyl quantization:} Initially the magnetic Weyl quantization is defined for Schwartz functions using the symplectic Fourier transform and magnetic Weyl system \cite{Mantoiu2004}. The symplectic Fourier transform for a Schwartz function $\phi\in\mathscr{S}(\R^{2d})$ is given by
\begin{equation*}
    \big(\mathcal{F}_\sigma \phi\big)(X)=\frac{1}{(2\pi)^d}\int \mathrm{d}Y\,e^{i\sigma(X,Y)}\phi(Y),
\end{equation*}
where $\sigma(X,Y)=y\cdot\xi-x\cdot\eta$ is the standard symplectic form \cite{Zworski2012}, and the magnetic Weyl system is defined as a family of $L^2(\R^d)$-unitaries $(w^A(X))_{X\in\R^{2d}}$ given by $w^A(X)=e^{-i\sigma\big(X,(Q,P^A)\big)}$, where $Q,P^A$ are respectively the position operators $Q=(Q_1,\dots,Q_d)$ and magnetic momentum operators $P^A=(P_1^A,\dots,P_d^A)$. Without going into to much depth we note that $\mathcal{F}_\sigma^{-1}=\mathcal{F}_\sigma$, and that $Q_j$, $P_j^A$, and $w^A(X)$ act on suitable $L^2(\R^d)$-functions $f$ in the following manner:
\begin{equation*}
    \big(Q_jf\big)(y)=y_jf(y),\qquad\big(P_j^Af\big)(y)=\big(-i\partial_j-A_j(y)\big)f(y),
\end{equation*}
\begin{equation*}
    \big (w^A(X)f\big ) (y)=e^{-i\left(y+\frac{x}{2}\right)\cdot\xi}e^{i\varphi(x+y,y)}f(x+y)
\end{equation*}
Then the magnetic Weyl quantization $\mathfrak{op}^A$ of a Schwartz function $\phi\in\mathscr{S}(\R^{2d})$ is given by
\begin{equation*}
    \langle g,\mathfrak{op}^A(\phi)f\rangle_{\mathscr{S}',\mathscr{S}}=\frac{1}{(2\pi)^d}\int \mathrm{d}X\,\big (\mathcal{F}_\sigma^{-1}\phi\big )(X)\, \langle g,w^A(X)f\rangle_{\mathscr{S}',\mathscr{S}}
\end{equation*}
for any Schwartz function $f\in\mathscr{S}(\R^d)$ and tempered distribution $g\in\mathscr{S}'(\R^d)$. We call $\phi\in\mathscr{S}(\R^{2d})$ a symbol and $\mathfrak{op}^A(\phi)$ a magnetic pseudo-differential operator.

Often one works with the kernel of $\mathfrak{op}^A(\phi)$ \cite{Mantoiu2004}:
\begin{equation}\label{hc4}
    \big(k^A\phi\big)(x,y)=\frac{1}{(2\pi)^d}\int \mathrm{d}\xi\, e^{i\xi\cdot(x-y)}e^{i\varphi(x,y)}\phi\bigg(\frac{x+y}{2},\xi\bigg)
\end{equation}
When the spaces involved are given their canonical topologies, which for spaces continuous operators is the topology of uniform convergence on bounded sets \cite{Osborne2014}, then the magnetic Weyl transform $k^A$ becomes a linear homeomorphism on $\mathscr{S}(\R^{2d})$, and $k^A$ is extendable to a linear homeomorphism of $\mathscr{S}'(\R^{2d})$. Thus the magnetic Weyl quantization is extendable to every tempered distribution $\phi\in\mathscr{S}'(\R^{2d})$, where $\mathfrak{op}^A(\phi)$ is then the unique map $\mathcal{B}(\mathscr{S}(\R^d),\mathscr{S}'(\R^d))$ with distributional kernel $k^A\phi$. With the kernel mapping\footnote{For $\phi\in\mathscr{S}'(\R^{2d})$ the map $\mathfrak{int}$ is defined by $\langle\mathfrak{int}(\phi)f,g\rangle_{\mathscr{S}',\mathscr{S}}=\langle\phi,g\otimes f\rangle_{\mathscr{S}',\mathscr{S}}$ for $f,g\in\mathscr{S}(\R^d)$.} $\mathfrak{int}\colon\mathscr{S}'(\R^{2d})\rightarrow\mathcal{B}(\mathscr{S}(\R^d),\mathscr{S}'(\R^d))$, which is a linear homeomorphism, the extension can be written as a composition of maps:
\begin{equation*}
    \mathfrak{op}^A=\mathfrak{int}\circ k^A
\end{equation*}

Notably, $\mathfrak{op}^A$ is a linear homeomorphism between several spaces of symbols and operators: $\mathfrak{op}^A$ is a linear homeomorphism of $\mathscr{S}'(\R^{2d})$ and $\mathcal{B}(\mathscr{S}(\R^d),\mathscr{S}'(\R^d))$, $\mathscr{S}(\R^{2d})$ and smoothing operators $\mathcal{B}(\mathscr{S}'(\R^d),\mathscr{S}(\R^d))$, and $L^2(\R^{2d})$ and Hilbert-Schmidt operators $\mathcal{B}_2(L^2(\R^d))$. See \cite{Mantoiu2004,Iftimie2007} for details.

\paragraph{Magnetic super Weyl quantization:} The magnetic super Weyl quantization is defined in analogy with the Weyl quantization \cite{LeeLein2022}: We define the symplectic Fourier transform $\mathcal{F}_\Sigma=\mathcal{F}_\sigma\otimes\mathcal{F}_\sigma$ and the super magnetic Weyl system $(W^A(\mathbf{X}))_{\mathbf{X}\in\R^{4d}}$ through
\begin{equation*}
    W^A(\mathbf{X})(T)=\big(w^A(X_L)\odot w^A(X_R)\big)(T)\coloneq w^A(X_L)Tw^A(X_R),
\end{equation*}
where $T$ is a suitable operator, e.g. take $T\in\mathcal{B}(L^2(\R^d))$. Then, if we define the duality bracket between $\mathcal{B}(\mathscr{S}(\R^d),\mathscr{S}'(\R^d))$ and $\mathcal{B}(\mathscr{S}'(\R^d),\mathscr{S}(\R^d))$ as
\begin{equation*}
    \langle R,T\rangle_{\mathcal{B}(\mathscr{S},\mathscr{S}'),\mathcal{B}(\mathscr{S}',\mathscr{S})}=\langle R,T\rangle_{\mathcal{B}',\mathcal{B}}\coloneq\langle\mathfrak{int}^{-1}(R),\mathfrak{int}^{-1}(T)\rangle_{\mathscr{S}',\mathscr{S}},
\end{equation*}
the magnetic super Weyl quantization $\mathfrak{Op}^A$ of a Schwartz function $\Phi\in\mathscr{S}(\R^{4d})$ is given by
\begin{equation}\label{eq:defn_super}
    \langle R,\mathfrak{Op}^A(\Phi)T\rangle_{\mathcal{B}',\mathcal{B}}=\frac{1}{(2\pi)^{2d}}\int \mathrm{d}\mathbf{X}\,\mathcal{F}_\Sigma^{-1}\Phi(\mathbf{X})\,\langle R,W^A(\mathbf{X})T\rangle_{\mathcal{B}',\mathcal{B}}
\end{equation}
for $T\in\mathcal{B}(\mathscr{S}'(\R^d),\mathscr{S}(\R^d))$ and $R\in\mathcal{B}(\mathscr{S}(\R^d),\mathscr{S}'(\R^d))$. This definition is such that for $\phi,\psi\in\mathscr{S}(\R^{2d})$ and $T\in\mathcal{B}(\mathscr{S}'(\R^d),\mathscr{S}(\R^d))$ we have:
\begin{equation}\label{eq:odot_identity}
    \mathfrak{Op}^A(\phi\otimes\psi)T=\mathfrak{op}^A(\phi)T\mathfrak{op}^A(\psi)=\big(\mathfrak{op}^A(\phi)\odot\mathfrak{op}^A(\psi)\big)(T)
\end{equation} 
We call $\Phi\in\mathscr{S}(\R^{2d})$ a super symbol and $\mathfrak{Op}^A(\Phi)$ a magnetic pseudo-differential super operator.

Similarly to before, we may extend the magnetic super Weyl quantization to tempered distributions through use of distributional kernels. In \cite{LeeLein2022} they find that
\begin{equation*}
    \mathfrak{Op}^A(\Phi)=\mathfrak{op}^A\circ\mathfrak{Int}(\tilde{K}^B\Phi)\circ\big ({\mathfrak{op}^A}\big )^{-1}
\end{equation*}
for some linear homeomorphism $\tilde{K}^B$ on $\mathscr{S}(\R^{4d})$ and $\mathfrak{Int}\colon\mathscr{S}'(\R^{4d})\rightarrow\mathcal{B}(\mathscr{S}(\R^{2d}),\mathscr{S}'(\R^{2d}))$ being the kernel map in a higher dimension. Note $\tilde{K}^B$ only depends on the magnetic field $B$, not the potential $A$. The map $\tilde{K}^B$ can then be extended and through the extension $\mathfrak{Op}^A$ is defined for every tempered distribution.

We will take another, but equivalent, approach. Instead of finding the symbol to symbol map $\mathfrak{Int}(\tilde{K}^B\Phi)$ we use the kernel to kernel map:
\begin{equation}\label{eq:kerneltokernel}
    \mathfrak{Op}^A(\Phi)=\mathfrak{int}\circ\mathfrak{Int}(K^A\Phi)\circ\mathfrak{int}^{-1},
\end{equation}
where $K^A$ is a magnetic super Weyl transform:
\begin{align}\label{eq:superweyltransform}
        \big(K^A\Phi\big)(x_L,x_R,y_L,y_R)&=\frac{1}{(2\pi)^{2d}}\int \mathrm{d}\xi_L\,\mathrm{d}\xi_R\, e^{i\Big[\xi_L\cdot(x_L-y_L)-\xi_R\cdot(x_R-y_R)\Big]}e^{i\Big[\varphi(x_L,y_L)-\varphi(x_R,y_R)\Big]}\nonumber \\
    &\quad\times\Phi\bigg(\frac{x_L+y_L}{2},\xi_L,\frac{x_R+y_R}{2},\xi_R\bigg)
\end{align}
Note $\tilde{K}^B$ and $K^A$ are closely related:
\begin{equation*}
    \mathfrak{Int}(K^A\Phi)=k^A\circ\mathfrak{Int}(\tilde{K}^B\Phi)\circ \big ({k^A}\big )^{-1}
\end{equation*}
The magnetic super Weyl transform $K^A$ is a linear homeomorphism of $\mathscr{S}(\R^{4d})$ and extendable to a linear homeomorphism of $\mathscr{S}'(\R^{4d})$. Using this, we may define the magnetic super Weyl quantization $\mathfrak{Op}^A(\Phi)$ of a tempered distribution $\Phi\in\mathscr{S}'(\R^{4d})$ as the map
\begin{equation*}
    \mathfrak{Op}^A(\Phi)=\mathfrak{int}\circ\mathfrak{Int}(K^A\Phi)\circ\mathfrak{int}^{-1}
\end{equation*}
in the space:
\begin{equation*}
    \mathcal{B}\Big(\mathcal{B}\big (\mathscr{S}'(\R^d),\mathscr{S}(\R^d)\big ),\mathcal{B}\big (\mathscr{S}(\R^d),\mathscr{S}'(\R^d)\big )\Big)
\end{equation*}
Note the identity \eqref{eq:odot_identity} still holds which can be proven by approximating a tensor product of tempered distribution with tensor products of Schwartz functions.

Similarly to $\mathfrak{op}^A$, $\mathfrak{Op}^A$ is a linear homeomorphism between several spaces of super symbols and super operators. Three of these pairs are: 
\begin{itemize}
\item The tempered distributions $\mathscr{S}'(\R^{4d})$ and
\begin{equation*}
    \mathcal{B}\Big(\mathcal{B}\big (\mathscr{S}'(\R^d),\mathscr{S}(\R^d)\big ),\mathcal{B}\big (\mathscr{S}(\R^d),\mathscr{S}'(\R^d)\big )\Big),
\end{equation*}
\item The Schwartz functions $\mathscr{S}(\R^{4d})$ and
\begin{equation*}
    \mathcal{B}\Big(\mathcal{B}\big (\mathscr{S}(\R^d),\mathscr{S}'(\R^d)\big ),\mathcal{B}\big (\mathscr{S}'(\R^d),\mathscr{S}(\R^d)\big )\Big),
\end{equation*}
\item $L^2(\R^{4d})$ and $\mathcal{B}_2\Big(\mathcal{B}_2\big (L^2(\R^d)\big )\Big)$.
\end{itemize}
Let us shortly explain the last pair: $L^2(\R^{4d})$ is isomorphic to $\mathcal{B}_2(L^2(\R^{2d}))$ which in turn is isomorphic to $\mathcal{B}_2\big(\mathcal{B}_2(L^2(\R^d))\big)$. The simplest example of an isomorphism $\mathcal{I}$ between $L^2(\R^{4d})$ and $\mathcal{B}_2\big(\mathcal{B}_2(L^2(\R^d))\big)$, maps $\Phi\in L^2(\R^{4d})$ to the super operator $\mathfrak{int}\circ\mathfrak{Int}(\Phi)\circ\mathfrak{int}^{-1}$. The magnetic super Weyl transform $K^A$ can be shown to be a linear homeomorphism of $L^2(\R^{4d})$, so $\mathfrak{Op}^A=\mathcal{I}\circ K^A$ is an isomorphism between $L^2(\R^{4d})$ and $\mathcal{B}_2\big(\mathcal{B}_2(L^2(\R^d))\big)$.

\begin{remark}
    Our definition in \eqref{eq:defn_super} does not entirely correspond to \cite[Definition IV.5]{LeeLein2022}, but the two are equal on smoothing operators. We also extend certain super operators to all $\mathcal{B}(\mathscr{S}(\R^d),\mathscr{S}'(\R^d))$-maps in Subsection \ref{sec:calc_semi}, which is equivalent to the extension in \cite{LeeLein2022}. Thus, in the end, both definitions will give the same magnetic pseudo-differential super operators for e.g. symbols in $\mathscr{S}(\R^{4d})$.
\end{remark}

\paragraph{Hörmander classes:} We call a function $m\colon\R^{n}\rightarrow(0,\infty)$, $n\in\N$, a tempered weight \cite{Nicola2010,Zworski2012,Thorn2026} if there exists $a,C>0$ such that
\begin{equation}\label{eq:peetre}
    m(u+v)\leq Cm(u)\langle v\rangle^a,\quad\forall u,v\in\R^n,
\end{equation}
where $\langle v\rangle=\sqrt{1+\Vert v\Vert^2}$. The ``usual" weight function on $\R^{2d}$ is $m(x,\xi)=\langle \xi\rangle^p$ with $x,\xi\in \R^d$ and $p\in \R$. 

Given a tempered weight $m$ on $\R^{2d}$ we define the Hörmander class
\begin{equation*}
    s_0(m)\coloneq\bigg\{\phi\in C^\infty(\R^{2d})\mid\sup_{X\in\R^{2d}}m(X)^{-1}|\partial^\gamma\phi(X)|<\infty,\forall\gamma\in\N_0^{2d}\bigg\}.
\end{equation*}
This is a Fréchet space with semi-norms given by:
\begin{equation*}
    \Vert\phi\Vert_{s_0(m),n}=\sum_{\gamma\in\N_0^{2d},|\gamma|\leq n}\sup_{X\in\R^{2d}}m(X)^{-1}|\partial^\gamma\phi(X)|
\end{equation*}
for $n\in\N_0$. Similarly we define the Hörmander super classes $S_0(M)\subseteq C^\infty(\R^{4d})$ for tempered weights $M$ over $\R^{4d}$. Lastly, we define the collection of such classes as $s_0(\infty)=\bigcup_ms_0(m)$ and $S_0(\infty)=\bigcup_MS_0(M)$. Note the intersections are just the Schwartz spaces over $\R^{2d}$ and $\R^{4d}$ respectively.

We collect some results on tempered weights in Appendix \ref{app:tempered_weights}. One of these, Lemma \ref{lem:app_smoorth_tempered}, implies that we may always assume that the tempered weights we are working with are smooth, since every Hörmander class $s_0(m)$ is isomorphic to another class $s_0(\tilde{m})$ where $\tilde{m}$ is smooth.

We refer to \cite{Thorn2026} for results on the magnetic Weyl quantizations of the Hörmander classes $s_0(m)$.

\section{Frame Decompositions}\label{sec:decompositions}

Originally from \cite{CorneanHelfferPurice2018,CorneanHelfferPurice2024} we consider the following frame: Let $\mathfrak{g}\colon\R^d\rightarrow\R$ be an element of $C_0^\infty(\R^d)$ such that
\begin{equation*}
    \supp(\mathfrak{g})\subseteq(-1,1)^d,\quad\sum_{\alpha\in\Z^d}(\tau_\alpha\mathfrak{g})^2\equiv1,
\end{equation*}
where $\big(\tau_yf\big)(x)=f(x-y)$ is a translation. Then for $\tilde{\alpha}=(\alpha,\alpha')\in\Z^{2d}$ we define the functions
\begin{equation*}
    \mathcal{G}_{\tilde{\alpha}}^A\colon\R^d\ni x\mapsto (2\pi)^{-\frac{d}{2}}e^{i\varphi(x,\alpha)}\mathfrak{g}(x-\alpha)e^{i\alpha'\cdot (x-\alpha)}.
\end{equation*}
Moreover, for $\tilde{\alpha},\tilde{\beta}\in\Z^{2d}$ we define the operators
\begin{equation*}
    \mathcal{T}_{\tilde{\alpha},\tilde{\beta}}^A\coloneq\mathfrak{int}(\mathcal{G}_{\tilde{\alpha}}^A\otimes \overline{\mathcal{G}_{\tilde{\beta}}^A}).
\end{equation*}
Note these operators were also used in \cite{LeeLein2025}.

The results \cite[Lemma 2.1, 2.2, and 2.3]{Thorn2026} show that an element $f$ from one of the spaces $L^2(\R^d)$, $\mathscr{S}(\R^d)$, or $\mathscr{S}'(\R^d)$ have a decomposition in $(\mathcal{G}_{\tilde{\alpha}}^A)_{\tilde{\alpha}\in\Z^d}$, i.e.:
\begin{equation*}
    f=\sum_{\tilde{\alpha}\in\Z^{2d}}\langle f,\overline{\mathcal{G}_{\tilde{\alpha}}^A}\rangle_{\mathscr{S}',\mathscr{S}}\mathcal{G}_{\tilde{\alpha}}^A
\end{equation*}
Conversely, if $(a_{\tilde{\alpha}})_{\tilde{\alpha}\in\Z^{2d}}$ are complex numbers satisfying intuitive conditions (given in the aforementioned lemmata), then
\begin{equation*}
    \sum_{\tilde{\alpha}\in\Z^{2d}}a_{\tilde{\alpha}}\mathcal{G}_{\tilde{\alpha}}^A
\end{equation*}
converges to an element of $L^2(\R^d)$, $\mathscr{S}(\R^d)$, or $\mathscr{S}'(\R^d)$.

Similar results can be obtained for the family $(\mathcal{T}_{\tilde{\alpha},\tilde{\beta}}^A)_{\tilde{\alpha},\tilde{\beta}\in\Z^{2d}}$ concerning the operator spaces $\mathcal{B}_2(L^2(\R^d))$, $\mathcal{B}(\mathscr{S}'(\R^d),\mathscr{S}(\R^d))$, and $\mathcal{B}(\mathscr{S}(\R^d),\mathscr{S}'(\R^d))$. To obtain these results we essentially only have to note that $\mathfrak{int}^{-1}$ is a linear homeomorphism between the above operator spaces and $L^2(\R^{2d})$, $\mathscr{S}(\R^{2d})$, and $\mathscr{S}'(\R^{2d})$, respectively, in addition to recalling \cite[Lemma 2.1, 2.2, and 2.3]{Thorn2026}:

\begin{lemma}\label{lem:convergence_l2}\cite[Proposition 3.2]{LeeLein2025}
    The family $(\mathcal{T}_{\tilde{\alpha},\tilde{\beta}}^A)_{\tilde{\alpha},\tilde{\beta}\in\Z^{2d}}$ defines a Parseval frame in $\mathcal{B}_2(L^2(\R^d))$ and\footnote{We define inner products to be antilinear in the first entry and linear in the second.}
    \begin{equation*}
        S=\sum_{\tilde{\alpha},\tilde{\beta}\in\Z^{2d}}\Big\langle\mathcal{T}_{\tilde{\alpha},\tilde{\beta}}^A,S\Big\rangle_{\mathcal{B}_2(L^2)}\mathcal{T}_{\tilde{\alpha},\tilde{\beta}}^A
    \end{equation*}
    holds for every $S\in\mathcal{B}_2(L^2(\R^d))$ with unconditional convergence.
\end{lemma}

\begin{proof}
    Using \cite[Lemma 2.1]{Thorn2026} we see that
    \begin{equation*}
        S=\sum_{\tilde{\alpha},\tilde{\beta}\in\Z^{2d}}\Big\langle\mathcal{T}_{\tilde{\alpha},\tilde{\beta}}^A,S\Big\rangle_{\mathcal{B}_2(L^2)}\mathcal{T}_{\tilde{\alpha},\tilde{\beta}}^A
    \end{equation*}
    with convergence in $\mathcal{B}_2(L^2(\R^d))$. Multiplying the above by $S^*$ and taking the trace proves that $(\mathcal{T}_{\tilde{\alpha},\tilde{\beta}}^A)_{\tilde{\alpha},\tilde{\beta}\in\Z^{2d}}$ is a Parseval frame.
\end{proof}

To state the frame decompositions in the cases of the operator spaces $\mathcal{B}(\mathscr{S}'(\R^d),\mathscr{S}(\R^d))$ and $\mathcal{B}(\mathscr{S}(\R^d),\mathscr{S}'(\R^d))$ we use the duality bracket introduced earlier:
\begin{equation*}
    \langle S, R\rangle_{\mathcal{B}',\mathcal{B}}=\langle \mathfrak{int}^{-1}(S), \mathfrak{int}^{-1}(R)\rangle_{\mathscr{S}',\mathscr{S}}
\end{equation*}
Note when $S,R$ are both smoothing operators, then by defining 
\begin{equation*}
    \mathfrak{int}^{-1}(\overline{S})\coloneq\overline{\mathfrak{int}^{-1}(S)}.
\end{equation*}
we get
\begin{equation*}
    \langle S, R\rangle_{\mathcal{B}_2(L^2)}=\langle R, \overline{S}\rangle_{\mathcal{B}',\mathcal{B}}.
\end{equation*}

\begin{lemma}\label{lem:convergence_schwartz}
    {\color{white}=}
    \begin{enumerate}[label={\rm(\roman*)}, ref={\rm(\roman*)}]
        \item\label{i_convergence_schwartz} If $S\in\mathcal{B}(\mathscr{S}'(\R^d),\mathscr{S}(\R^d))$, then for any $n\in\N_0$ we have
        \begin{equation*}
            \sup_{\tilde{\alpha},\tilde{\beta}\in\Z^{2d}}\langle(\tilde{\alpha},\tilde{\beta})\rangle^{n}\Big|\Big\langle S,\overline{\mathcal{T}_{\tilde{\alpha},\tilde{\beta}}^A}\Big\rangle_{\mathcal{B}',\mathcal{B}}\Big|<\infty.
        \end{equation*}
        \item\label{ii_convergence_schwartz} If $(N_{\tilde{\alpha},\tilde{\beta}})_{\tilde{\alpha},\tilde{\beta}\in\Z^{2d}}$ are complex numbers such that
        \begin{equation*}
            \sup_{\tilde{\alpha},\tilde{\beta}\in\Z^{2d}}\langle(\tilde{\alpha},\tilde{\beta})\rangle^{n}|N_{\tilde{\alpha},\tilde{\beta}}|<\infty
        \end{equation*}
        for all $n\in\N_0$, then $\sum_{\tilde{\alpha},\tilde{\beta}\in\Z^{2d}}N_{\tilde{\alpha},\tilde{\beta}}\mathcal{T}_{\tilde{\alpha},\tilde{\beta}}^A$ converges absolutely in $\mathcal{B}(\mathscr{S}'(\R^d),\mathscr{S}(\R^d))$.
    \end{enumerate}
\end{lemma}

\begin{lemma}\label{lem:convergence_tempered}
    {\color{white}=}
    \begin{enumerate}[label={\rm(\roman*)}, ref={\rm(\roman*)}]
        \item\label{i_convergence_tempered} If $S\in\mathcal{B}(\mathscr{S}(\R^d),\mathscr{S}'(\R^d))$, then for some $n\in\N_0$ we have
        \begin{equation*}
            \sup_{\tilde{\alpha},\tilde{\beta}\in\Z^{2d}}\langle(\tilde{\alpha},\tilde{\beta})\rangle^{-n}\Big|\Big\langle S,\overline{\mathcal{T}_{\tilde{\alpha},\tilde{\beta}}^A}\Big\rangle_{\mathcal{B}',\mathcal{B}}\Big|<\infty.
        \end{equation*}
        \item\label{ii_convergence_tempered} If $(N_{\tilde{\alpha},\tilde{\beta}})_{\tilde{\alpha},\tilde{\beta}\in\Z^{2d}}$ are complex numbers for which there exists $n\in\N_0$ such that
        \begin{equation*}
            \sup_{\tilde{\alpha},\tilde{\beta}\in\Z^{2d}}\langle(\tilde{\alpha},\tilde{\beta})\rangle^{-n}|N_{\tilde{\alpha},\tilde{\beta}}|<\infty,
        \end{equation*}
        then $\sum_{\tilde{\alpha},\tilde{\beta}\in\Z^{2d}}N_{\tilde{\alpha},\tilde{\beta}}\mathcal{T}_{\tilde{\alpha},\tilde{\beta}}^A$ converges absolutely in $\mathcal{B}(\mathscr{S}(\R^d),\mathscr{S}'(\R^d))$.
    \end{enumerate}
\end{lemma}

As we stated before both Lemma \ref{lem:convergence_schwartz} and \ref{lem:convergence_tempered} follow from \cite[Lemma 2.2]{Thorn2026} and \cite[Lemma 2.3]{Thorn2026}, in addition to Schwartz' kernel theorem \cite{Treves2006}.

Using Lemma \ref{lem:convergence_l2} and Lemma \ref{lem:convergence_schwartz} we can conclude that
\begin{equation}\label{eq:matrix_decomp}
    S=\sum_{\tilde{\alpha},\tilde{\beta}\in\Z^{2d}}\Big\langle S,\overline{\mathcal{T}_{\tilde{\alpha},\tilde{\beta}}^A}\Big\rangle_{\mathcal{B}',\mathcal{B}}\mathcal{T}_{\tilde{\alpha},\tilde{\beta}}^A
\end{equation}
holds for every smoothing operator $S$ with convergence in $\mathcal{B}(\mathscr{S}'(\R^d),\mathscr{S}(\R^d))$, and using Lemma \ref{lem:convergence_tempered} together with duality we conclude that this expansion holds for every $S\in\mathcal{B}(\mathscr{S}(\R^d),\mathscr{S}'(\R^d))$.

\subsection{Matrix Representation of Super Operators}

Using Lemma \ref{lem:convergence_schwartz} and Lemma \ref{lem:convergence_tempered} we can make a matrix representation of every super operator:
\begin{equation}\label{eq:supermatrix}
    \begin{aligned}
        \mathfrak{Op}^A(\Phi)S&=\sum_{(\tilde{\alpha},\tilde{\beta}),(\tilde{\gamma},\tilde{\delta})\in\Z^{4d}}\Big\langle \mathfrak{Op}^A(\Phi)\mathcal{T}_{\tilde{\gamma},\tilde{\delta}}^A,\overline{\mathcal{T}_{\tilde{\alpha},\tilde{\beta}}^A}\Big\rangle_{\mathcal{B}',\mathcal{B}}\,\Big\langle S,\overline{\mathcal{T}_{\tilde{\gamma},\tilde{\delta}}^A}\Big\rangle_{\mathcal{B}',\mathcal{B}}\mathcal{T}_{\tilde{\alpha},\tilde{\beta}}^A\\
        &=\sum_{(\tilde{\alpha},\tilde{\beta}),(\tilde{\gamma},\tilde{\delta})\in\Z^{4d}}\mathbb{M}_{(\tilde{\alpha},\tilde{\beta}),(\tilde{\gamma},\tilde{\delta})}^A[ \mathfrak{Op}^A(\Phi)]\,\big(\mathcal{T}_{\tilde{\alpha},\tilde{\beta}}^A\bowtie\overline{\mathcal{T}_{\tilde{\gamma},\tilde{\delta}}^A}\big) S
    \end{aligned}
\end{equation}
where
\begin{equation*}
    \mathbb{M}^A[\mathfrak{Op}^A(\Phi)]=\Big(\Big\langle \mathfrak{Op}^A(\Phi)\mathcal{T}_{\tilde{\gamma},\tilde{\delta}}^A,\overline{\mathcal{T}_{\tilde{\alpha},\tilde{\beta}}^A}\Big\rangle_{\mathcal{B}',\mathcal{B}}\Big)_{(\tilde{\alpha},\tilde{\beta}),(\tilde{\gamma},\tilde{\delta})\in\Z^{4d}}
\end{equation*}
and $\mathcal{T}_{\tilde{\alpha},\tilde{\beta}}^A\bowtie\overline{\mathcal{T}_{\tilde{\gamma},\tilde{\delta}}^A}$ is the map
\begin{equation*}
    \mathcal{B}\big(\mathscr{S}(\R^d),\mathscr{S}'(\R^d)\big)\ni S\mapsto \Big\langle S,\overline{\mathcal{T}_{\tilde{\gamma},\tilde{\delta}}^A}\Big\rangle_{\mathcal{B}',\mathcal{B}}\mathcal{T}_{\tilde{\alpha},\tilde{\beta}}^A\in \mathcal{B}\big(\mathscr{S}'(\R^d),\mathscr{S}(\R^d)\big).
\end{equation*}
Lee and Lein also presented such an expansion in \cite[Section 3.2]{LeeLein2025}, but for super operators on $\mathcal{B}_2(L^2(\R^d))$. Note that they used the notation $\mathcal{T}_{\tilde{\alpha},\tilde{\gamma}}^A\odot\mathcal{T}_{\tilde{\delta},\tilde{\beta}}^A=\mathcal{T}_{\tilde{\alpha},\tilde{\beta}}^A\bowtie\overline{\mathcal{T}_{\tilde{\gamma},\tilde{\delta}}^A}$ in their expansion.

Now we present a matrix characterization for our super symbol classes $S_0(M)$, an extension of \cite[Theorem 4.1 and Corollary 4.2]{LeeLein2025}.

\begin{theorem}\label{thm:matrix_rep}
    Fix a tempered weight $M$. Then:
    \begin{enumerate}[label={\rm(\roman*)}, ref={\rm(\roman*)}]
        \item\label{i_matrix_rep} For any $n\in\N_0$ there exists $C>0,k\in\N_0$ such that for all $\Phi\in S_0(M)$:
        \begin{equation}\label{eq:matrix_specific}
        \begin{aligned}
            \sup_{(\tilde{\alpha},\tilde{\beta}),(\tilde{\gamma},\tilde{\delta})\in\Z^{4d}}&\langle(\tilde{\alpha},\tilde{\beta})-(\tilde{\gamma},\tilde{\delta})\rangle^{n}M\Bigg(\frac{(\tilde{\alpha},\tilde{\beta})+(\tilde{\gamma},\tilde{\delta})}{2}\Bigg)^{-1}\Big|\mathbb{M}_{(\tilde{\alpha},\tilde{\beta}),(\tilde{\gamma},\tilde{\delta})}^A[\mathfrak{Op}^A(\Phi)]\Big|\\
            &< C\Vert\Phi\Vert_{S_0(M),k}
        \end{aligned}
        \end{equation}
        \item\label{ii_matrix_rep} If $(N_{(\tilde{\alpha},\tilde{\beta}),(\tilde{\gamma},\tilde{\delta})})_{(\tilde{\alpha},\tilde{\beta}),(\tilde{\gamma},\tilde{\delta})\in\Z^{4d}}$ are complex numbers such that
        \begin{equation}\label{eq:matrix}
            \begin{aligned}
                \sup_{(\tilde{\alpha},\tilde{\beta}),(\tilde{\gamma},\tilde{\delta})\in\Z^{4d}}&\langle(\tilde{\alpha},\tilde{\beta})-(\tilde{\gamma},\tilde{\delta})\rangle^{n}M\Bigg(\frac{(\tilde{\alpha},\tilde{\beta})+(\tilde{\gamma},\tilde{\delta})}{2}\Bigg)^{-1}\Big|N_{(\tilde{\alpha},\tilde{\beta}),(\tilde{\gamma},\tilde{\delta})}\Big|< \infty
            \end{aligned}
        \end{equation}
        for all $n\in\N_0$, then
        \begin{equation}\label{eq:dequant_sum}
            \sum_{(\tilde{\alpha},\tilde{\beta}),(\tilde{\gamma},\tilde{\delta})\in\Z^{4d}}N_{(\tilde{\alpha},\tilde{\beta}),(\tilde{\gamma},\tilde{\delta})}\big ({\mathfrak{Op}^A}\big )^{-1}\Big(\mathcal{T}_{\tilde{\alpha},\tilde{\beta}}^A\bowtie\overline{\mathcal{T}_{\tilde{\gamma},\tilde{\delta}}^A}\Big)
        \end{equation}
        converges uniformly on compacts towards a member of $S_0(M)$ and each semi-norm of the limit is bounded by a constant times \eqref{eq:matrix} for some $n$.
    \end{enumerate}
\end{theorem}

We keep the proof short since versions of this result has already appeared: \cite[Theorem 3.1]{CorneanHelfferPurice2024}, \cite[Theorem 4.1]{LeeLein2025}, and \cite[Theorem 2.4]{Thorn2026}.

\begin{proof}[Proof of \ref{i_matrix_rep}]
    Using the kernel to kernel mapping \eqref{eq:kerneltokernel}, we see that an element of the matrix of $\mathfrak{Op}^A(\Phi)$ is given by:
    \begin{align*}
        \mathbb{M}_{(\tilde{\alpha},\tilde{\beta}),(\tilde{\gamma},\tilde{\delta})}^A[\mathfrak{Op}^A(\Phi)]
        &=\big\langle \mathfrak{Int}(K^A\Phi)\mathcal{G}_{\tilde{\gamma}}^A\otimes\overline{\mathcal{G}_{\tilde{\delta}}^A},\overline{\mathcal{G}_{\tilde{\alpha}}^A}\otimes\mathcal{G}_{\tilde{\beta}}^A\big\rangle_{\mathscr{S}',\mathscr{S}}\\
        &=\big\langle K^A\Phi,\mathcal{G}_{\tilde{\gamma}}^A\otimes\overline{\mathcal{G}_{\tilde{\delta}}^A}\otimes\overline{\mathcal{G}_{\tilde{\alpha}}^A}\otimes\mathcal{G}_{\tilde{\beta}}^A\big\rangle_{\mathscr{S}',\mathscr{S}}
    \end{align*}
    This last integral can be estimated in exactly the same manner as $\big\langle k^A\phi,\overline{\mathcal{G}_{\tilde{\alpha}}^A}\otimes\mathcal{G}_{\tilde{\beta}}^A\big\rangle_{\mathscr{S}',\mathscr{S}}$ with $\phi\in s_0(m)$ for some tempered weight $m$ on $\R^{2d}$, which was already done in the proof of \cite[Theorem 2.4]{Thorn2026}. The essential steps are as follows:
    \begin{enumerate}
        \item Perform a coordinate change in the integral
        \begin{equation*}
            \big\langle K^A\Phi,\mathcal{G}_{\tilde{\gamma}}^A\otimes\overline{\mathcal{G}_{\tilde{\delta}}^A}\otimes\overline{\mathcal{G}_{\tilde{\alpha}}^A}\otimes\mathcal{G}_{\tilde{\beta}}^A\big\rangle_{\mathscr{S}',\mathscr{S}}.
        \end{equation*}
        \item Apply partial integration using the exponential factors in the aforementioned integral.
        \item Take the absolute value, and repeatedly use Peetre's inequality and that $\Phi$ and its derivatives are bounded when multiplied by $M^{-1}$.\qedhere
    \end{enumerate}
\end{proof}

\begin{proof}[Proof of \ref{ii_matrix_rep}]
    Let us consider a single term in the sum first. By the dequantization results of \cite{LeeLein2022}, or see \eqref{eq:odot_identity}, we have:
    \begin{equation*}
        \big ({\mathfrak{Op}^A}\big )^{-1}\Big(\mathcal{T}_{\tilde{\alpha},\tilde{\beta}}^A\bowtie\overline{\mathcal{T}_{\tilde{\gamma},\tilde{\delta}}^A}\Big)=\big ({\mathfrak{Op}^A}\big )^{-1}\Big(\mathcal{T}_{\tilde{\alpha},\tilde{\gamma}}^A\odot\mathcal{T}_{\tilde{\delta},\tilde{\beta}}^A\Big)=\big ({\mathfrak{op}^A}\big )^{-1}(\mathcal{T}_{\tilde{\alpha},\tilde{\gamma}}^A)\otimes\big ({\mathfrak{op}^A}\big )^{-1}(\mathcal{T}_{\tilde{\delta},\tilde{\beta}}^A)
    \end{equation*}
    In the proof of \cite[Theorem 2.4]{Thorn2026} the function $\big ({\mathfrak{0p}^A}\big )^{-1}(\mathcal{T}_{\tilde{\alpha},\tilde{\gamma}}^A)$ was analyzed and we got the upper bound:
    \begin{equation*}
        \big|\partial^{\omega_1}\big ({\mathfrak{op}^A}\big )^{-1}(\mathcal{T}_{\tilde{\alpha},\tilde{\gamma}}^A)(X)\big|\leq C_11_{\frac{\alpha+\gamma}{2}+(-1,1)^d}(x)\bigg\langle\xi-\frac{\alpha'+\gamma'}{2}\bigg\rangle^{-k_1}\langle\alpha-\gamma\rangle^{k_2}\langle\alpha'-\gamma'\rangle^{k_3}
    \end{equation*}
    for all $x,\xi\in\R^d$ with $C_1,k_1,k_2,k_3>0$ where $C_1,k_2,k_3$ depend on $\omega_1\in\N_0^{2d}$ and $C_1$ also depends on $k_1$, which can be chosen arbitrary large. Combining this with \eqref{eq:matrix} and using Peetre's inequality we get:
    \begin{align}\label{eq:matrix_2main_estimate}
            \Big|N_{(\tilde{\alpha},\tilde{\beta}),(\tilde{\gamma},\tilde{\delta})}&\partial^{\omega_2}\big ({\mathfrak{Op}^A}\big )^{-1}\Big(\mathcal{T}_{\tilde{\alpha},\tilde{\beta}}^A\bowtie\overline{\mathcal{T}_{\tilde{\gamma},\tilde{\delta}}^A}\Big)(X,Y)\Big|\nonumber\\
        &\leq C_21_{\frac{1}{2}\left(\alpha+\gamma,\beta+\delta\right)+(-1,1)^{2d}}(x,y)\bigg(\bigg\langle\xi-\frac{\alpha'+\gamma'}{2}\bigg\rangle\bigg\langle\eta-\frac{\beta'+\delta'}{2}\bigg\rangle\bigg)^{-a-d-1}\nonumber\\
        &\quad \times M\Bigg(\frac{(\tilde{\alpha},\tilde{\beta})+(\tilde{\gamma},\tilde{\delta})}{2}\Bigg)\big(\langle\alpha-\gamma\rangle\langle\alpha'-\gamma'\rangle\langle\beta-\delta\rangle\langle\beta'-\delta'\rangle\big)^{-d-1}\\
        &\leq C_31_{\frac{1}{2}\left(\alpha+\gamma,\beta+\delta\right)+(-1,1)^{2d}}(x,y) M(X,Y)\nonumber\\
        &\quad\times\bigg(\bigg\langle\xi-\frac{\alpha'+\gamma'}{2}\bigg\rangle\bigg\langle\eta-\frac{\beta'+\delta'}{2}\bigg\rangle\langle\alpha-\gamma\rangle\langle\alpha'-\gamma'\rangle\langle\beta-\delta\rangle\langle\beta'-\delta'\rangle\bigg)^{-d-1}\nonumber
        \end{align}
    for all $X,Y\in\R^{2d}$ with $C_2,C_3>0$ depending on $\omega_2\in\N_0^{4d}$ and $a>0$ stemming from \eqref{eq:peetre} for $M$. The bound \eqref{eq:matrix_2main_estimate} is uniform in $(\tilde{\alpha},\tilde{\beta}),(\tilde{\gamma},\tilde{\delta})\in\Z^{4d}$ and shows that the series
    \begin{equation*}
        \sum_{(\tilde{\alpha},\tilde{\beta}),(\tilde{\gamma},\tilde{\delta})\in\Z^{4d}}N_{(\tilde{\alpha},\tilde{\beta}),(\tilde{\gamma},\tilde{\delta})}\big ({\mathfrak{Op}^A}\big )^{-1}\Big(\mathcal{T}_{\tilde{\alpha},\tilde{\beta}}^A\bowtie\overline{\mathcal{T}_{\tilde{\gamma},\tilde{\delta}}^A}\Big)
    \end{equation*}
    and its term-wise derivatives converge uniformly on compacts, thus the limit is a $C^\infty(\R^{4d})$-function. From \eqref{eq:matrix_2main_estimate} it also follows that all partial sums are uniformly bounded when multiplied with $M^{-1}(X,Y)$, which also extends to the sums of term-wise derivatives. This implies that the sum converges unconditionally to a member of $S_0(M)$. The statement on the semi-norms of the limit follows from the logic behind \eqref{eq:matrix_2main_estimate}.
\end{proof}

\begin{remark}
    We will in the following repeat the phrase ``tracking the relevant estimate ...", or something similar, when e.g. proving continuity of a product. What we in general allude to is the tedious process of using mainly Theorem \ref{thm:matrix_rep}, where in \ref{i_matrix_rep} we bound the matrix elements of $\mathfrak{Op}^A(\Phi)$ by, among other things, a $S_0(M)$-semi-norm of $\Phi$, and in \ref{ii_matrix_rep} state that $S_0(M)$-semi-norm of the sum \eqref{eq:dequant_sum} is bounded by, among other things, factors of the kind \eqref{eq:matrix}. The same strategy was used in \cite{Thorn2026}.
\end{remark}

\section{Super Calculus of Hörmander Symbols}\label{sec:calc}

With most of the definitions and the frame decomposition out of the way, we begin proving useful results for the magnetic super Weyl quantization. First we are interested in the magnetic semi-super Weyl product and super Weyl product, where the former concerns itself with products of the kind;
\begin{equation*}
    \mathfrak{Op}^A(\Phi)\mathfrak{op}^A(\psi)=\mathfrak{op}^A(\Phi\mathbin{\bullet^B}\psi)
\end{equation*}
and the latter with products of the kind:
\begin{equation*}
    \mathfrak{Op}^A(\Phi)\mathfrak{Op}^A(\Psi)=\mathfrak{Op}^A(\Phi\mathbin{\#^B}\Psi)
\end{equation*}
The definitions of these products and their corresponding Moyal algebras stem from \cite{LeeLein2022} were the authors also proved several results for some classical Hörmander super classes, see \cite[Lemma V.5, Lemma V.9, Proposition VI.4, and Proposition VI.5]{LeeLein2022}. We extend these to include the Hörmander super classes $S_0(M)$.

The first step is to use Theorem \ref{thm:matrix_rep} to prove the following:

\begin{lemma}\label{lem:invariance_smoothing}
    For $\Phi\in S_0(\infty)$, the magnetic pseudo-differential super operator $\mathfrak{Op}^A(\Phi)$ has range in the smoothing operators and it is bounded as an operator on the space of smoothing operators, that is:
    \begin{equation*}
        \mathfrak{Op}^A(\Phi)\in\mathcal{B}\Big(\mathcal{B}\big(\mathscr{S}'(\R^d),\mathscr{S}(\R^d)\big)\Big)
    \end{equation*}
    Moreover, for any tempered weight $M$, $\mathfrak{Op}^A$ maps $S_0(M)$ continuously into the super operator space $\mathcal{B}\big(\mathcal{B}(\mathscr{S}'(\R^d),\mathscr{S}(\R^d))\big)$.
\end{lemma}

\begin{proof}
    Let $S\in\mathcal{B}\left(\mathscr{S}'(\R^d),\mathscr{S}(\R^d)\right)$. Using Lemma \ref{lem:convergence_schwartz} \ref{i_convergence_schwartz} and Theorem \ref{thm:matrix_rep} \ref{i_matrix_rep} we see that
    \begin{equation*}
        N_{\tilde{\alpha},\tilde{\beta}}=\sum_{(\tilde{\gamma},\tilde{\beta})\in\Z^{4d}}\Big\langle \mathfrak{Op}^A(\Phi)\mathcal{T}_{\tilde{\gamma},\tilde{\delta}}^A,\overline{\mathcal{T}_{\tilde{\alpha},\tilde{\beta}}^A}\Big\rangle_{\mathcal{B}',\mathcal{B}}\,\Big\langle S,\overline{\mathcal{T}_{\tilde{\gamma},\tilde{\delta}}^A}\Big\rangle_{\mathcal{B}',\mathcal{B}}
    \end{equation*}
    for $\tilde{\alpha},\tilde{\beta}\in\Z^{2d}$ defines a collection of complex numbers satisfying Lemma \ref{lem:convergence_schwartz} \ref{ii_convergence_schwartz}, whence Lemma \ref{lem:convergence_schwartz} \ref{ii_convergence_schwartz} and \eqref{eq:supermatrix} tells us that $\mathfrak{Op}^A(\Phi)S$ is a smoothing operator. In more detail, using Theorem \ref{thm:matrix_rep} \ref{i_matrix_rep}, Peetre's inequality, and \eqref{eq:peetre} we get
    \begin{equation*}
         \Big|\Big\langle \mathfrak{Op}^A(\Phi)\mathcal{T}_{\tilde{\gamma},\tilde{\delta}}^A,\overline{\mathcal{T}_{\tilde{\alpha},\tilde{\beta}}^A}\Big\rangle_{\mathcal{B}',\mathcal{B}}\Big|\leq C\langle (\tilde{\alpha},\tilde{\beta})\rangle^{-n}\langle (\tilde{\gamma},\tilde{\delta})\rangle^k
    \end{equation*}
    for arbitrary $n\in\N_0$ with $C,k>0$ depending on $n$. So multiplying $N_{\tilde{\alpha},\tilde{\beta}}$ by $\langle (\tilde{\alpha},\tilde{\beta})\rangle$ for some $l\in\N_0$, we absorb $\langle (\tilde{\alpha},\tilde{\beta})\rangle^l$ into $|\langle \mathfrak{Op}^A(\Phi)\mathcal{T}_{\tilde{\gamma},\tilde{\delta}}^A,\overline{\mathcal{T}_{\tilde{\alpha},\tilde{\beta}}^A}\rangle_{\mathcal{B}',\mathcal{B}}|$ generating $\langle (\tilde{\gamma},\tilde{\delta})\rangle^{k}$ for some $k\in\N_0$ depending on $l$, which is then absorbed into $|\langle S,\overline{\mathcal{T}_{\tilde{\gamma},\tilde{\delta}}^A}\rangle_{\mathcal{B}',\mathcal{B}}|$ using Lemma \ref{lem:convergence_schwartz} \ref{i_convergence_schwartz}, leaving an absolute convergent sum, uniformly in $(\tilde{\alpha},\tilde{\beta})\in\Z^{4d}$.
    
    Tracking the relevant estimates we get the last two assertions on continuity.
\end{proof}

\subsection{Magnetic Semi-super Weyl Product}\label{sec:calc_semi}

The magnetic semi-super Weyl product between a tempered distribution $\Phi\in\mathscr{S}'(\R^{4d})$ and a Schwartz function $\psi\in\mathscr{S}(\R^{2d})$ is defined as the tempered distribution:
\begin{equation}\label{eq:semi_weyl}
    \Phi\mathbin{\bullet^B}\psi=\big ({\mathfrak{op}^A}\big )^{-1}\big(\mathfrak{Op}^A(\Phi)\mathfrak{op}^A(\psi)\big)=\big ({k^A}\big )^{-1}\mathfrak{Int}(K^A\Phi)k^A\psi=\mathfrak{Int}(\tilde{K}^B\Phi)\psi
\end{equation}
The question is then, for which $\Phi$ can we extend the above formula to every tempered distribution $\psi$, that is, does $\Phi$ belong to the magnetic semi-super Moyal space $\mathfrak{m}^B(\R^{4d})$ \cite[Definition V.3]{LeeLein2022}. Here we work with the formal transpose \cite{Grubb2009,ReedSimonI1980}, i.e. given $\Phi\in\mathscr{S}'(\R^{4d})$ is there a map $T_\Phi\in\mathcal{B}(\mathscr{S}(\R^{2d}))$ such that
\begin{equation*}
    \langle \Phi\mathbin{\bullet^B}\psi,\chi\rangle_{\mathscr{S}',\mathscr{S}}=\langle\psi,T_\Phi\chi\rangle_{\mathscr{S}',\mathscr{S}}
\end{equation*}
for all $\psi,\chi\in\mathscr{S}(\R^{2d})$. If the answer is affirmative, then $\Phi\in\mathfrak{m}^B(\R^{4d})$ and the map $\mathscr{S}(\R^{2d})\ni\psi\mapsto\Phi\mathbin{\bullet^B}\psi$ is extendable to a map in $\mathcal{B}(\mathscr{S}'(\R^{2d}))$.

The maps $\big ({k^A}\big )^{-1},k^A$ pose no problems since they are known to have simple formal transposes. So considering \eqref{eq:semi_weyl}, $\Phi\in \mathfrak{m}^B(\R^{4d})$ if and only if $\mathfrak{Int}(K^A\Phi)$ has a formal transpose. Here we note that a formal transpose of $\mathfrak{Int}(K^A\Phi)$ would have kernel $(K^A\Phi)^t=\overline{K^A\overline{\Phi}}$, where $\Psi^t(X_L,X_R)=\Psi(X_R,X_L)$ for Schwartz functions which is then extended to tempered distributions.

\begin{proposition}\label{prop:calc_super_semi}
    The space of Hörmander super classes $S_0(\infty)$ is contained in the magnetic semi-super Moyal space $\mathfrak{m}^B(\R^{4d})$.
    
    Furthermore, for any tempered weights $M$ and $m$ on $\R^{4d}$ and $\R^{2d}$ respectively, $S_0(M)\mathbin{\bullet^B}s_0(m)\subseteq s_0(\tilde{m})$ where
    \begin{equation*}
        \tilde{m}\colon\R^{2d}\ni X\mapsto M(X,X)m(X),
    \end{equation*}
    and the bilinear map $S_0(M)\times s_0(m)\ni(\Phi,\psi)\mapsto\Phi\mathbin{\bullet^B}\psi\in s_0(\tilde{m})$ is continuous.
\end{proposition}

\begin{proof}
    By the above analysis, we need to prove that $\mathfrak{Int}(\overline{K^A\overline{\Phi}})\in\mathcal{B}(\mathscr{S}(\R^{2d}))$ for $\Phi\in S_0(\infty)$. But
    \begin{equation*}
        \mathfrak{Int}(\overline{K^A\overline{\Phi}})\psi=\overline{\mathfrak{Int}(K^A\overline{\Phi})\overline{\psi}}
    \end{equation*}
    and $\overline{\Phi}\in S_0(\infty)$ implies that $\mathfrak{Int}(K^A\overline{\Phi})\in \mathcal{B}(\mathscr{S}(\R^{2d}))$ by Lemma \ref{lem:invariance_smoothing}. Thus $\Phi\in\mathfrak{m}^B(\R^{4d})$.

    To prove that $S_0(M)\mathbin{\bullet^B}s_0(m)\subseteq s_0(\tilde{m})$ we study the coordinates of $\mathfrak{Op}^A(\Phi)\mathfrak{op}^A(\psi)$ in the frame $(\mathcal{T}_{\tilde{\alpha},\tilde{\beta}}^A)_{\tilde{\alpha},\tilde{\beta}\in\Z^{2d}}$ for $\Phi\in S_0(M)$ and $\psi\in s_0(m)$. The result \cite[Theorem 2.4]{Thorn2026}, an analog of Theorem \ref{thm:matrix_rep} for the classes $s_0(m)$, would then imply that $\Phi\mathbin{\bullet^B}\psi\in s_0(\tilde{m})$, and tracking estimates would prove continuity.

    We have
    \begin{align*}
        \Big\langle\mathfrak{Op}^A&(\Phi)\mathfrak{op}^A(\psi),\overline{\mathcal{T}_{\tilde{\alpha},\tilde{\beta}}^A}\Big\rangle_{\mathcal{B}',\mathcal{B}}\\
        &=\sum_{(\tilde{\gamma},\tilde{\delta})\in\Z^{4d}}\Big\langle\mathfrak{Op}^A(\Phi)\mathcal{T}_{\tilde{\gamma},\tilde{\delta}}^A,\overline{\mathcal{T}_{\tilde{\alpha},\tilde{\beta}}^A}\Big\rangle_{\mathcal{B}',\mathcal{B}}\,\Big\langle\mathfrak{op}^A(\psi),\overline{\mathcal{T}_{\tilde{\gamma},\tilde{\delta}}^A}\Big\rangle_{\mathcal{B}',\mathcal{B}}
    \end{align*}
    using the decomposition \eqref{eq:matrix_decomp} on $\mathfrak{op}^A(\psi)$. Theorem \ref{thm:matrix_rep} \ref{i_matrix_rep} gives the estimate
    \begin{equation}\label{eq:semi_matrix_est}
        \Big|\Big\langle\mathfrak{Op}^A(\Phi)\mathcal{T}_{\tilde{\gamma},\tilde{\delta}}^A,\overline{\mathcal{T}_{\tilde{\alpha},\tilde{\beta}}^A}\Big\rangle_{\mathcal{B}',\mathcal{B}}\Big|\leq C\Big(\langle\tilde{\alpha}-\tilde{\gamma}\rangle\langle\tilde{\beta}-\tilde{\delta}\rangle\Big)^{-n}M\Bigg(\frac{(\tilde{\alpha},\tilde{\beta})+(\tilde{\gamma},\tilde{\delta})}{2}\Bigg)
    \end{equation}
    and \cite[Theorem 2.4 1]{Thorn2026} gives the estimate
    \begin{equation}\label{eq:semi_coor_est}
        \Big|\Big\langle\mathfrak{op}^A(\psi),\overline{\mathcal{T}_{\tilde{\gamma},\tilde{\delta}}^A}\Big\rangle_{\mathcal{B}',\mathcal{B}}\Big|\leq \tilde{C}\langle\tilde{\gamma}-\tilde{\delta}\rangle^{-\tilde{n}}m\Bigg(\frac{\tilde{\gamma}+\tilde{\delta}}{2}\Bigg),
    \end{equation}
    where $n,\tilde{n}\in\N_0$ are arbitrary but $C,\tilde{C}>0$ depends on $n,\tilde{n}$ respectively. From \eqref{eq:peetre} we get
    \begin{align*}
        M&\Bigg(\frac{(\tilde{\alpha},\tilde{\beta})+(\tilde{\gamma},\tilde{\delta})}{2}\Bigg)m\Bigg(\frac{\tilde{\gamma}+\tilde{\delta}}{2}\Bigg)\\
        &\leq C_1M\Bigg(\frac{(\tilde{\alpha}+\tilde{\beta},\tilde{\alpha}+\tilde{\beta})}{2}\Bigg)m\Bigg(\frac{\tilde{\alpha}+\tilde{\beta}}{2}\Bigg)\big(\langle \tilde{\beta}-\tilde{\gamma}\rangle\langle \tilde{\alpha}-\tilde{\delta}\rangle\big)^{a_1+a_2}
    \end{align*}
    with $C_1>0$ and $a_1,a_2>0$ determined by $M,m$ respectively. Using Peetre's inequality we get
    \begin{equation*}
        \big(\langle \tilde{\beta}-\tilde{\gamma}\rangle\langle \tilde{\alpha}-\tilde{\delta}\rangle\big)^{a_1+a_2}\leq C_2\big(\langle\tilde{\beta}-\tilde{\delta}\rangle\langle\tilde{\delta}-\tilde{\gamma}\rangle\langle\tilde{\alpha}-\tilde{\gamma}\rangle\langle\tilde{\gamma}-\tilde{\delta}\rangle\big)^{a_1+a_2}
    \end{equation*}
    for some $C_2>0$, so by \eqref{eq:semi_matrix_est} and \eqref{eq:semi_coor_est}:
    \begin{align*}
        \Big|\Big\langle\mathfrak{Op}^A&(\Phi)\mathfrak{op}^A(\psi),\overline{\mathcal{T}_{\tilde{\alpha},\tilde{\beta}}^A}\Big\rangle_{\mathcal{B}',\mathcal{B}}\Big|\\
        &\leq C_3\tilde{m}\Bigg(\frac{\tilde{\alpha}+\tilde{\beta}}{2}\Bigg)\sum_{(\tilde{\gamma},\tilde{\delta})\in\Z^{4d}}\big(\langle\tilde{\alpha}-\tilde{\gamma}\rangle\langle\tilde{\gamma}-\tilde{\delta}\rangle\langle\tilde{\delta}-\tilde{\beta}\rangle\big)^{-n}
    \end{align*}
    with $n\in\N_0$ arbitrary, $C_3>0$ depending on $n$, and $\tilde{m}$ defined as in the proposition. Thus, using Peetre's inequality again, we obtain:
    \begin{align*}
        \sup_{\tilde{\alpha},\tilde{\beta}\in\Z^{2d}}\langle\tilde{\alpha}-\tilde{\beta}\rangle^{\tilde{n}}\tilde{m}\Bigg(\frac{\tilde{\alpha}+\tilde{\beta}}{2}\Bigg)^{-1}\Big|\Big\langle\mathfrak{Op}^A&(\Phi)\mathfrak{op}^A(\psi),\overline{\mathcal{T}_{\tilde{\alpha},\tilde{\beta}}^A}\Big\rangle_{\mathcal{B}',\mathcal{B}}\Big|\\
        &\leq C_4\sum_{(\tilde{\gamma},\tilde{\delta})\in\Z^{4d}}\langle(\tilde{\gamma},\tilde{\delta})\rangle^{-4d-1}<\infty
    \end{align*}
    for some $C_4>0$. This estimate and \cite[Theorem 2.4 2]{Thorn2026} then implies that $\Phi\mathbin{\bullet^B}\psi\in s_0(\tilde{m})$, and we are done.
\end{proof}

\begin{remark}
    In \cite{Thorn2026} several results were proven that related properties of the tempered weight $m$ to properties of $\mathfrak{op}^A(s_0(m))$, e.g. if $m\in L^1(\R^{2d})$, then $\mathfrak{op}^A(s_0(m))$ has a continuous embedding into $\mathcal{B}_1(L^2(\R^d))$ \cite[Theorem 4.3]{Thorn2026}.
    
    One example of the applications of the above result is classifying how certain super operators $\mathfrak{Op}^A(\Phi)$ behaves on spaces $\mathfrak{op}^A(s_0(m))$, e.g. if $m\in L^1(\R^{2d})$ and $\Phi\in S_0(1)$, then $\mathfrak{Op}^A(\Phi)\mathfrak{op}^A(s_0(m))\subseteq\mathfrak{op}^A(s_0(m))\subseteq \mathcal{B}_1(L^2(\R^d))$, showing that the super operator $\mathfrak{Op}^A(\Phi)$ is invariant on some subspaces of trace-class operators.
\end{remark}

\subsection{Magnetic Super Weyl Product}

Much like the magnetic Moyal product \cite{Mantoiu2004,Iftimie2007}, the magnetic super Weyl product \cite{LeeLein2022} exist between any tempered distribution $\Phi\in\mathscr{S}'(\R^{4d})$ and Schwartz function $\Psi\in\mathscr{S}(\R^{4d})$:
\begin{equation*}
    \Phi\mathbin{\#^B}\Psi=\big ({\mathfrak{Op}^A}\big )^{-1}\big(\mathfrak{Op}^A(\Phi)\mathfrak{Op}^A(\Psi)\big)
\end{equation*}
and
\begin{equation*}
    \Psi\mathbin{\#^B}\Phi=\big ({\mathfrak{Op}^A}\big )^{-1}\big(\mathfrak{Op}^A(\Psi)\mathfrak{Op}^A(\Phi)\big)
\end{equation*}
Note that it is often the case that $\Phi\mathbin{\#^B}\Psi \neq\Psi\mathbin{\#^B}\Phi$. Again we have the question of when the magnetic super Weyl product between two tempered distributions can be defined, i.e. which tempered distributions lie in the magnetic super Moyal algebra $\mathscr{M}^B(\R^{4d})$ \cite[Definition V.8]{LeeLein2022}:

\begin{proposition}\label{prop:calc_super_moyal}
    The space of Hörmander super classes $S_0(\infty)$ is contained in the magnetic super Moyal space $\mathscr{M}^B(\R^{4d})$.
    
    Furthermore, for any two tempered weights $M_1$ and $M_2$, $S_0(M_1)\mathbin{\#^B}S_0(M_2)\subseteq S_0(M_1M_2)$ and the bilinear map $S_0(M_1)\times S_0(M_2)\ni(\Phi,\Psi)\mapsto\Phi\mathbin{\#^B}\Psi\ni S_0(M_1M_2)$ is continuous.
\end{proposition}

\begin{proof}
    Let $\Phi\in S_0(M_1)$ and $\Psi\in S_0(M_2)$ for two tempered weights $M_1,M_2$. By Lemma \ref{lem:invariance_smoothing}, the product $\mathfrak{Op}^A(\Phi)\mathfrak{Op}^A(\Psi)$ is a well-defined super operator. Now, if we can prove that $\big ({\mathfrak{Op}^A}\big )^{-1}\big(\mathfrak{Op}^A(\Phi)\mathfrak{Op}^A(\Psi)\big)\in S_0(M_1M_2)$, then it automatically follows  that $\Phi,\Psi\in\mathscr{M}^B(\R^{4d})$ by the following reasoning:

    We have $\mathscr{S}(\R^{4d})=\bigcap_{M}S_0(M)$, so if $\Phi\in S_0(M_1)$ and $\Psi\in \mathscr{S}(\R^{4d})$, then $\Psi\in S_0(M_1^{-1}M)$ for any tempered weight $M$, so
    \begin{equation*}
        \Phi\#^B\Psi=\big ({\mathfrak{Op}^A}\big )^{-1}\big(\mathfrak{Op}^A(\Phi)\mathfrak{Op}^A(\Psi)\big)\in S_0(M).
    \end{equation*}
    Since $M$ is arbitrary, $\Phi\#^B\Psi\in\bigcap_{M}S_0(M)=\mathscr{S}(\R^{4d})$, and similar reasoning would prove $\Psi\#^B\Phi\in\mathscr{S}(\R^{4d})$, whence $\Phi\in\mathscr{M}^B(\R^{4d})$.

    Thus we need to prove that $\big ({\mathfrak{Op}^A}\big )^{-1}\big(\mathfrak{Op}^A(\Phi)\mathfrak{Op}^A(\Psi)\big)\in S_0(M_1M_2)$ for $\Phi\in S_0(M_1)$ and $\Psi\in S_0(M_2)$. By Theorem \ref{thm:matrix_rep} \ref{ii_matrix_rep} it is enough to study the matrix elements of $\mathfrak{Op}^A(\Phi)\mathfrak{Op}^A(\Psi)$ in the frame $(\mathcal{T}_{\tilde{\alpha},\tilde{\beta}}^A)_{\tilde{\alpha},\tilde{\beta}\in\Z^{2d}}$. These are given by
    \begin{align*}
        \mathbb{M}_{(\tilde{\alpha},\tilde{\beta}),(\tilde{\gamma},\tilde{\delta})}^A[&\mathfrak{Op}^A(\Phi)\mathfrak{Op}^A(\Psi)]=\Big\langle \mathfrak{Op}^A(\Phi)\mathfrak{Op}^A(\Psi)\mathcal{T}_{\tilde{\gamma},\tilde{\delta}}^A,\overline{\mathcal{T}_{\tilde{\alpha},\tilde{\beta}}^A}\Big\rangle_{\mathcal{B}',\mathcal{B}}\\
        &=\sum_{(\tilde{\mu},\tilde{\nu})\in\Z^{4d}}\Big\langle \mathfrak{Op}^A(\Psi)\mathcal{T}_{\tilde{\gamma},\tilde{\delta}}^A,\overline{\mathcal{T}_{\tilde{\mu},\tilde{\nu}}^A}\Big\rangle_{\mathcal{B}',\mathcal{B}}\,\Big\langle \mathfrak{Op}^A(\Phi)\mathcal{T}_{\tilde{\mu},\tilde{\nu}}^A,\overline{\mathcal{T}_{\tilde{\alpha},\tilde{\beta}}^A}\Big\rangle_{\mathcal{B}',\mathcal{B}}
    \end{align*}
    using the decomposition \eqref{eq:matrix_decomp} on $\mathfrak{Op}^A(\Psi)\mathcal{T}_{\tilde{\gamma},\tilde{\delta}}^A$. Using Theorem \ref{thm:matrix_rep} \ref{i_matrix_rep}, Peetre's inequality, and \eqref{eq:peetre}, we obtain that
    \begin{equation*}
        \Big|\mathbb{M}_{(\tilde{\alpha},\tilde{\beta}),(\tilde{\gamma},\tilde{\delta})}^A[\mathfrak{Op}^A(\Phi)\mathfrak{Op}^A(\Psi)]\Big|\leq C \big(M_1M_2\big)\Bigg(\frac{(\tilde{\alpha},\tilde{\beta})+(\tilde{\gamma},\tilde{\beta})}{2}\Bigg)\langle(\tilde{\alpha},\tilde{\beta})-(\tilde{\gamma},\tilde{\beta})\rangle^{-n}
    \end{equation*}
    with $n\in\N_0$ arbitrary large and $C>0$ depending on $n$. This shows that $\mathfrak{Op}^A(\Phi)\mathfrak{Op}^A(\Psi)$ has a super symbol in $S_0(M_1M_2)$ by Theorem \ref{thm:matrix_rep} \ref{ii_matrix_rep}.

    The continuity assertion again follows from tracking the relevant estimates.
\end{proof}

\section{Boundedness}\label{sec:bound}

The second big topic, beside the calculi, is continuity and regularity of super operators. By Proposition \ref{prop:calc_super_semi}, $\mathfrak{Op}^A(\Phi)$ is extendable to the whole of $\mathcal{B}(\mathscr{S}(\R^d),\mathscr{S}'(\R^d))$ when $\Phi\in S_0(\infty)$. Since spaces of bounded operators and Schatten-classes between different Lebesgue spaces and magnetic Sobolev spaces are included in $\mathcal{B}(\mathscr{S}(\R^d),\mathscr{S}'(\R^d))$, a natural question is when $\mathfrak{Op}^A(\Phi)$ possesses some extra properties like boundedness when restricted to such spaces.

We will give three results. First we give criteria for the tempered weight $M$ ensuring boundedness of $\mathfrak{Op}^A(\Phi)$, $\Phi\in S_0(M)$, between bounded operators and trace-class operators. This will be based on simple techniques and we assume one can obtain more comprehensive results, meaning assumptions on $M$ leading to boundedness between other spaces of operators. Secondly, we consider super symbols that are limits of tensor products of Hörmander symbols, in a sense extending \cite[Lemma IV.6]{LeeLein2022}. The third type of results are mainly on Schatten-class properties of the super operators $\mathfrak{Op}^A(\Phi)$ when restricted to Hilbert-Schmidt operators between magnetic Sobolev spaces, in a way extending \cite[Theorem 5.3 and Corollary 5.8]{LeeLein2025}.

Our main focus will be on operators between the magnetic Sobolev spaces $H_A^s(\R^d)$ from \cite[Definition 3.4 and 3.10]{Iftimie2007}. Based on the discussion in \cite[Section 5.2]{LeeLein2025}, which in turn is based on results from \cite{Mantoiu2007,Iftimie2010}, we can give a simple description of these spaces: 

\begin{definition}\label{defhc1}
There exists a family of functions $(\chi_s)_{s\in\R}$ in $s_0(\infty)$ such that we can identify $H_A^s(\R^d)=\mathfrak{op}^A(\chi_{-s})L^2(\R^d)$, $s\in\R$, with norm
\begin{equation*}
    \Vert f\Vert_{H_A^s}\coloneq\Vert\mathfrak{op}^A(\chi_s) f\Vert_{L^2},\quad f\in H_A^s(\R^d),
\end{equation*}
and $H_A^{-s}(\R^d)$ is the anti-dual of $H_A^s(\R^d)$ with the duality bracket:
\begin{equation*}
    \langle g,f\rangle_{H_A^{-s},H_A^{s}}\coloneq\langle\mathfrak{op}^A(\chi_{-s})g,\mathfrak{op}^A(\chi_s)f\rangle_{L^2},\quad f\in H_A^s(\R^d),g\in H_A^{-s}(\R^d).
\end{equation*}
Moreover, $\chi_s\in s_0(m_0^s)$ with the tempered weight 
\begin{equation}\label{hc2}
m_0\colon \R^{2d}\ni X=(x,\xi)\mapsto\langle\xi\rangle
\end{equation}
and $\chi_s\mathbin{\star^B}\chi_{-s}=1$ with $\star^B$ being the magnetic Moyal product \cite{Mantoiu2004,Iftimie2007}.
\end{definition}
Other spaces one could consider are magnetic Sobolev spaces of $p$-integrable functions and other fractional magnetic Sobolev space \cite{Nguyen2018,Nguyen2020}, magnetic analogs of the Sobolev space $H(m)$ with $m$ a tempered weight on $\R^{2d}$, see \cite[Definition 1.5.2]{Nicola2010} or \cite[Section 8.3]{Zworski2012}, or any of these but having weighted Lebesgue spaces $L_w^p(\R^d)=L^p(\R^d,w(x)dx)$ with $w$ a tempered weight on $\R^d$.

\subsection{Boundedness on Operator Spaces I}

First we present a simple result on boundedness.

\begin{proposition}\label{prop:opspaces}
    Let $s_L,s_R,s_L',s_R'\in\R$. If $M$ is a tempered weight such that
    \begin{equation*}
        \int \mathrm{d}\mathbf{X}\,M(\mathbf{X})m_0^{s_L'-s_L}(X_L)m_0^{s_R-s_R'}(X_R)<\infty,
    \end{equation*}
    that is $M(m_0^{s_L'-s_L}\otimes m_0^{s_R-s_R'})\in L^1(\R^{4d})$, then $\mathfrak{Op}^A(\Phi)$ with $\Phi\in S_0(M)$ is an operator in
    \begin{equation*}
        \mathcal{B}\Big(\mathcal{B}\big(H_A^{s_R}(\R^d),H_A^{s_L}(\R^d)\big),\mathcal{B}_1\big(H_A^{s_R'}(\R^d),H_A^{s_L'}(\R^d)\big)\Big)
    \end{equation*}
    and $\mathfrak{Op}^A$ maps $S_0(M)$ continuously into this space of super operators.
\end{proposition}

\begin{proof}
    The magnetic pseudo-differential super operator
    $$\mathfrak{op}^A(\chi_{-s_L})\odot\mathfrak{op}^A(\chi_{s_R})=\mathfrak{Op}^A(\chi_{-s_L}\otimes \chi_{s_R})$$
    is a linear homeomorphism of $\mathcal{B}(L^2(\R^d))$ into $\mathcal{B}\big(H_A^{s_R}(\R^d),H_A^{s_L}(\R^d)\big)$ and of $\mathcal{B}_1(L^2(\R^d))$ into $\mathcal{B}_1\big(H_A^{s_R}(\R^d),H_A^{s_L}(\R^d)\big)$, and its inverse is given by $\mathfrak{op}^A(\chi_{s_L})\odot\mathfrak{op}^A(\chi_{-s_R})$. So
    \begin{equation*}
        \mathfrak{Op}^A(\Phi)\in\mathcal{B}\Big(\mathcal{B}\big(H_A^{s_R}(\R^d),H_A^{s_L}(\R^d)\big),\mathcal{B}_1\big(H_A^{s_R'}(\R^d),H_A^{s_L'}(\R^d)\big)\Big)
    \end{equation*}
    if and only if
    \begin{equation*}
        \Big(\mathfrak{op}^A(\chi_{s_L'})\odot\mathfrak{op}^A(\chi_{-s_R'})\Big)\mathfrak{Op}^A(\Phi)\Big(\mathfrak{op}^A(\chi_{-s_L})\odot\mathfrak{op}^A(\chi_{s_R})\Big)\in\mathcal{B}\Big(\mathcal{B}\big(L^2(\R^d)\big),\mathcal{B}_1\big(L^2(\R^d)\big)\Big).
    \end{equation*}
    But if $\Phi\in S_0(M)$ with $M(m_0^{s_L'-s_L}\otimes m_0^{s_R-s_R'})\in L^1(\R^{4d})$, then
    \begin{align*}
        \Big(\mathfrak{op}^A(\chi_{s_L'})&\odot\mathfrak{op}^A(\chi_{-s_R'})\Big)\mathfrak{Op}^A(\Phi)\Big(\mathfrak{op}^A(\chi_{-s_L})\odot\mathfrak{op}^A(\chi_{s_R})\Big)\\
        &=\mathfrak{Op}^A((\chi_{s_L'}\otimes \chi_{-s_R'})\mathbin{\#^B}\Phi\mathbin{\#^B}(\chi_{-s_L}\otimes\chi_{s_R}))=\mathfrak{Op}^A(\tilde{\Phi})
    \end{align*}
    with $\tilde{\Phi}\in S_0(\tilde{M})$ where $\tilde{M}\in L^1(\R^{4d})$. Hence the proposition is proven if we can show that $M\in L^1(\R^{4d})$ implies that:
    \begin{equation*}
        \mathfrak{Op}^A(S_0(M))\hookrightarrow\mathcal{B}\Big(\mathcal{B}\big(L^2(\R^d)\big),\mathcal{B}_1\big(L^2(\R^d)\big)\Big)
    \end{equation*}
    
    Let $\Phi\in S_0(M)$ and $S\in\mathcal{B}(L^2(\R^d))$. Using Lemma \ref{lem:convergence_tempered} we expand $\mathfrak{Op}^A(\Phi)S$ similarly to \eqref{eq:supermatrix}, getting:
    \begin{equation}\label{eq:supermatrix_imposter}
        \begin{aligned}
            \mathfrak{Op}^A(\Phi)S&=\sum_{(\tilde{\alpha},\tilde{\beta}),(\tilde{\gamma},\tilde{\delta})\in\Z^{4d}}\mathbb{M}_{(\tilde{\alpha},\tilde{\beta}),(\tilde{\gamma},\tilde{\delta})}^A[ \mathfrak{Op}^A(\Phi)]\big(\mathcal{T}_{\tilde{\alpha},\tilde{\beta}}^A\bowtie\overline{\mathcal{T}_{\tilde{\gamma},\tilde{\delta}}^A}\big) S
        \end{aligned}
    \end{equation}
    Here
    \begin{align*}
        \big\Vert\big(\mathcal{T}_{\tilde{\alpha},\tilde{\beta}}^A\bowtie\overline{\mathcal{T}_{\tilde{\gamma},\tilde{\delta}}^A}\big) S\big\Vert_{\mathcal{B}_1(L^2)}=\Vert\mathcal{T}_{\tilde{\alpha},\tilde{\beta}}^A\Vert_{\mathcal{B}_1(L^2)}|\langle S\mathcal{G}_{\tilde{\delta}}^A,\overline{\mathcal{G}_{\tilde{\gamma}}^A}\rangle_{L^2}|\leq\Vert\mathcal{G}_0^A\Vert_{L^2}^4\Vert S\Vert_{\mathcal{B}(L^2)}
    \end{align*}
    and $\mathbb{M}^A[\mathfrak{Op}^A(\Phi)]\in\ell^1(\Z^{8d})$ by the assumption $M\in L^1(\R^{4d})$ and Theorem \ref{thm:matrix_rep}. Thus the sum in \eqref{eq:supermatrix_imposter} is absolutely convergent in $\mathcal{B}_1(L^2(\R^d))$, and we have
    \begin{equation*}
        \Vert\mathfrak{Op}^A(\Phi)S\Vert_{\mathcal{B}_1(L^2(\R^d))}\leq C\Vert \Phi\Vert_{S_0(M),k}\Vert S\Vert_{\mathcal{B}(L^2(\R^d))}
    \end{equation*}
    with $C>0,k\in\N_0$.
\end{proof}

\begin{remark}\label{rem:app1}
    Note that we needed the condition $\mathbb{M}^A[\mathfrak{Op}^A(\Phi)]\in\ell^1(\Z^{8d})$ and assumed the stronger condition $M\in L^1(\R^{4d})$. But this is equivalent for tempered weights up to measurability, see Lemma \ref{lem:app_lebesgue}, which can then be fixed by using Lemma \ref{lem:app_smoorth_tempered}.
\end{remark}

\begin{remark}
    The above crude method of proving boundedness extends in a trivial manner to other spaces, e.g.:
    \begin{itemize}
        \item Suppose $q_L,q_R,q_L',q_R'\in[1,\infty)$. Then, if $M\in L^1(\R^{4d})$ is a tempered weight, then $\mathfrak{Op}^A(\Phi)$ is an operator in $\mathcal{B}\big(\mathcal{B}(L^{q_R}(\R^d),L^{q_L}(\R^d)),\mathcal{B}(L^{q_R'}(\R^d),L^{q_L'}(\R^d))\big)$ and $\mathfrak{Op}^A$ maps $S_0(M)$ continuously into $\mathcal{B}\big(\mathcal{B}(L^{q_R}(\R^d),L^{q_L}(\R^d)),\mathcal{B}(L^{q_R'}(\R^d),L^{q_L'}(\R^d))\big)$.
        \item For Hilbert spaces $V,W$ we have the continuous injections:
        \begin{equation*}
            \mathcal{B}_q(V,W)\hookrightarrow\mathcal{B}_{q'}(V,W)\hookrightarrow\mathcal{B}(V,W)
        \end{equation*}
        with $0<q<q'\leq\infty$. Thus the above result also shows that, as an example, if $M\in L^1(\R^{4d})$, then $\mathfrak{Op}^A(\Phi)$ lies in $\mathcal{B}\big(\mathcal{B}_q\big(H_A^s(\R^d)\big)\big)$ for any $q\in(0,\infty]$ and $s\in\R$.
    \end{itemize}
\end{remark}

\subsection{Boundedness on Operator Spaces II}

Next we extend \cite[Lemma IV.6]{LeeLein2022} by considering tensor products of Hörmander classes.

We take two tempered weights $m_1,m_2$ on $\R^{2d}$ and consider $s_0(m_1)\mathbin{\otimes} s_0(m_2)$. We equip $s_0(m_1)\mathbin{\otimes} s_0(m_2)$ with the projective topology \cite{Treves2006}, i.e. the topology induced by the semi-norms $\Vert\cdot\Vert_{s_0(m_1),n}\otimes\Vert\cdot\Vert_{s_0(m_2),k}$ for $n,k\in\N_0$, and denote the completion of $s_0(m_1)\mathbin{\otimes} s_0(m_2)$ by $s_0(m_1)\mathbin{\widehat{\otimes}} s_0(m_2)$. Note that $s_0(m_1)\mathbin{\widehat{\otimes}} s_0(m_2)\hookrightarrow S_0(m_1\otimes m_2)$.

\begin{proposition}\label{prop:continuoustensor}
    Let $s_L,s_R,s_L',s_R'\in\R$ and let $m_0$ be as in \eqref{hc2}. 
    Then for every super symbol $\Phi\in s_0(m_0^{s_L-s_L'})\mathbin{\widehat{\otimes}} s_0(m_0^{s_R'-s_R})$, the magnetic pseudo-differential super operator $\mathfrak{Op}^A(\Phi)$ is a bounded operator from $\mathcal{B}(H_A^{s_R}(\R^d),H_A^{s_L}(\R^d))$ to $\mathcal{B}(H_A^{s_R'}(\R^d),H_A^{s_L'}(\R^d))$, and from $\mathcal{B}_p(H_A^{s_R}(\R^d),H_A^{s_L}(\R^d))$ to $\mathcal{B}_p(H_A^{s_R'}(\R^d),H_A^{s_L'}(\R^d))$ where $p\in[1,\infty]$.
\end{proposition}

\begin{proof}
    We only prove the proposition for the Schatten-classes $\mathcal{B}_p(H_A^{s_R}(\R^d),H_A^{s_L}(\R^d))$ and $\mathcal{B}_p(H_A^{s_R'}(\R^d),H_A^{s_L'}(\R^d))$, but the proof of boundedness between spaces of bounded operators is the same. 
    
    By \cite[Theorem 45.1]{Treves2006} we may write $\Phi$ as an absolutely convergent sum $\sum_{n=0}^\infty\lambda_n \phi_n\otimes\psi_n$, where $(\lambda_n)_{n\in\N_0}$ are absolutely summable complex numbers, and $(\phi_n)_{n\in\N_0}$ and $(\psi_n)_{n\in\N_0}$ are sequences in $s_0(m_0^{s_L-s_L'})$ and $s_0(m_0^{s_R'-s_R})$ respectively, with both sequences converging to zero in these spaces. Let $\Phi_N\coloneq\sum_{n=0}^N\lambda_n\phi_n\otimes\psi_n$ for $N\in\N_0$. By \cite[Proposition 3.14]{Iftimie2007} we see that $\mathfrak{op}^A(\phi_n)\in \mathcal{B}(H_A^{s_L}(\R^d),H_A^{s_L'}(\R^d))$ and $\mathfrak{op}^A(\psi_n)\in\mathcal{B}(H_A^{s_R'}(\R^d),H_A^{s_R}(\R^d))$ for all $n\in\N_0$, whence
    \begin{equation*}
        \mathfrak{Op}^A(\Phi_N)S=\sum_{n=0}^N\lambda_n\mathfrak{op}^A(\phi_n)S\mathfrak{op}^A(\psi_n)
    \end{equation*}
    is in $\mathcal{B}_p(H_A^{s_R'}(\R^d),H_A^{s_L'}(\R^d))$ whenever $S\in\mathcal{B}_p(H_A^{s_R}(\R^d),H_A^{s_L}(\R^d))$. Furthermore, it is easy to see that the super operators $\mathfrak{Op}^A(\Phi_N)$ are bounded and they form a Cauchy sequence in 
    \begin{equation*}
        \mathcal{B}\Big(\mathcal{B}_p\big(H_A^{s_R}(\R^d),H_A^{s_L}(\R^d)\big),\mathcal{B}_p\big(H_A^{s_R'}(\R^d),H_A^{s_L'}(\R^d)\big)\Big)
    \end{equation*}
    since for some $k\in\N_0,C>0$ we have
    \begin{align*}
        \Vert \mathfrak{Op}^A(\Phi_{N_1})&-\mathfrak{Op}^A(\Phi_{N_2})\Vert_{\mathcal{B}\big(\mathcal{B}_p(H_A^{s_R},H_A^{s_L}),\mathcal{B}_p(H_A^{s_R'},H_A^{s_L'})\big)}\\
        &\leq\sum_{n=\min\{N_1,N_2\}+1}^\infty|\lambda_n|\Vert\mathfrak{op}^A(\phi_n)\Vert_{\mathcal{B}(H_A^{s_L},H_A^{s_L'})}\Vert\mathfrak{op}^A(\psi_n)\Vert_{\mathcal{B}(H_A^{s_R'},H_A^{s_R})}\\
        &\leq C\sum_{n=\min\{N_1,N_2\}+1}^\infty|\lambda_n|\Vert \phi_n\Vert_{s_0(1),k}\Vert \psi_n\Vert_{s_0(1),k}
    \end{align*}
    which goes to zero when $N_1,N_2\rightarrow\infty$. The super operators $\mathfrak{Op}^A(\Phi_N)$ already converge to $\mathfrak{Op}^A(\Phi)$ in the space
    \begin{equation*}
        \mathcal{B}\Big(\mathcal{B}\big(\mathscr{S}'(\R^d),\mathscr{S}(\R^d)\big),\mathcal{B}\big(\mathscr{S}(\R^d),\mathscr{S}'(\R^d)\big)\Big),
    \end{equation*}
    so we conclude that $\mathfrak{Op}^A(\Phi)$ is a bounded super operator from $\mathcal{B}_p(H_A^{s_R}(\R^d),H_A^{s_L}(\R^d))$ to $\mathcal{B}_p(H_A^{s_R'}(\R^d),H_A^{s_L'}(\R^d))$.
\end{proof}

\begin{remark}
    Proposition \ref{prop:continuoustensor} could be generalized to boundedness between Schatten-classes using \cite[Theorem 4.1, 4.2, and 4.3, and Remark 4.4]{Thorn2026}. E.g. the super quantization of a super symbol $\Phi\in s_0(m)\mathbin{\widehat{\otimes}}s_0(1)$ with $m\in L^p(\R^{2d})$ is bounded from $\mathcal{B}(L^2(\R^d))$ to $\mathcal{B}_p(L^2(\R^d))$ for every $p\in[1,\infty)$.
\end{remark}

\subsection{Schatten-class Properties on Hilbert-Schmidt Operators}

We now extend \cite[Theorem 5.3 and Corollary 5.8]{LeeLein2025} to Schatten-class properties on Hilbert-Schmidt operators, similarly to how the magnetic Calderón-Vaillancourt theorem \cite[Theorem 3.1]{Iftimie2007} was extended in \cite{Thorn2026}.

\begin{theorem}\label{thm:caldvail}
    Let $s_L,s_R,s_L',s_R'\in\R$. If $M$ is a tempered weight in $\R^{4d}$ satisfying one of the following assumptions:
    \begin{enumerate}
        \item $M(m_0^{s_L'-s_L}\otimes m_0^{s_R-s_R'})\in L^\infty(\R^{4d})$
        \item $\sup_{\Vert\mathbf{X}\Vert>R}M(m_0^{s_L'-s_L}\otimes m_0^{s_R-s_R'})(\mathbf{X})\stackrel{R\rightarrow\infty}{\rightarrow}0$
        \item $M(m_0^{s_L'-s_L}\otimes m_0^{s_R-s_R'})\in L^p(\R^{4d})$ with $p\in(0,\infty)$
    \end{enumerate}
    Then for $\Phi\in S_0(M)$ we have correspondingly:
    \begin{enumerate}
        \item $\mathfrak{Op}^A(\Phi)\in \mathcal{B}\Big(\mathcal{B}_2\big(H_A^{s_R}(\R^d),H_A^{s_L}(\R^d)\big),\mathcal{B}_2\big(H_A^{s_R'}(\R^d),H_A^{s_L'}(\R^d)\big)\Big)$
        \item $\mathfrak{Op}^A(\Phi)\in \mathcal{B}_\infty\Big(\mathcal{B}_2\big(H_A^{s_R}(\R^d),H_A^{s_L}(\R^d)\big),\mathcal{B}_2\big(H_A^{s_R'}(\R^d),H_A^{s_L'}(\R^d)\big)\Big)$
        \item $\mathfrak{Op}^A(\Phi)\in \mathcal{B}_p\Big(\mathcal{B}_2\big(H_A^{s_R}(\R^d),H_A^{s_L}(\R^d)\big),\mathcal{B}_2\big(H_A^{s_R'}(\R^d),H_A^{s_L'}(\R^d)\big)\Big)$
    \end{enumerate}
    In every case, the magnetic super Weyl quantization $\mathfrak{Op}^A$ maps $S_0(M)$ continuously into the super operator space.
\end{theorem}

\begin{remark}
    Note that the first case, i.e. when $M(m_0^{s_L'-s_L}\otimes m_0^{s_R-s_R'})\in L^\infty(\R^{4d})$, is exactly \cite[Corollary 5.8]{LeeLein2025}, which is restated here for completeness.
\end{remark}

\begin{proof}
    We may reduce the statement to that of operator spaces on $L^2(\R^d)$ using the same methods as in the proof of Proposition \ref{prop:opspaces}. Also note that the smoothing operators $\mathcal{B}(\mathscr{S}'(\R^d),\mathscr{S}(\R^d))$ are dense in $\mathcal{B}_2(L^2(\R^d))$.

    \paragraph{Proof of 1.} We use the Schur test similarly to the proof of \cite[Theorem 5.3]{LeeLein2025}, originally based on the proof of \cite[Theorem 3.7]{CorneanHelfferPurice2024}.

    Using Lemma \ref{lem:convergence_l2} and Lemma \ref{lem:convergence_schwartz} we get the following expansion for smoothing operators $S$:
    \begin{equation*}
        \begin{aligned}
            \Vert\mathfrak{Op}^A(\Phi)S\Vert_{\mathcal{B}_2(L^2)}^2&=\sum_{(\tilde{\alpha},\tilde{\beta})\in\Z^{4d}}\Big\langle\mathfrak{Op}^A(\Phi)S,\mathcal{T}_{\tilde{\alpha},\tilde{\beta}}^A\Big\rangle_{\mathcal{B}_2(L^2)}\Big\langle\mathcal{T}_{\tilde{\alpha},\tilde{\beta}}^A,\mathfrak{Op}^A(\Phi)S\Big\rangle_{\mathcal{B}_2(L^2)}\\
            &=\sum_{(\tilde{\alpha},\tilde{\beta}),(\tilde{\gamma},\tilde{\delta}),(\tilde{\mu},\tilde{\nu})\in\Z^{4d}}\overline{\Big\langle\mathcal{T}_{\tilde{\gamma},\tilde{\delta}}^A,S\Big\rangle_{\mathcal{B}_2(L^2)}\mathbb{M}_{(\tilde{\alpha},\tilde{\beta}),(\tilde{\gamma},\tilde{\delta})}^A[ \mathfrak{Op}^A(\Phi)]}\\
            &\qquad\times\mathbb{M}_{(\tilde{\alpha},\tilde{\beta}),(\tilde{\mu},\tilde{\nu})}^A[ \mathfrak{Op}^A(\Phi)]\Big\langle \mathcal{T}_{\tilde{\mu},\tilde{\nu}}^A,S\Big\rangle_{\mathcal{B}_2(L^2)}
        \end{aligned}
    \end{equation*}
    Lemma \ref{lem:convergence_l2} tells us that $\big(\big\langle \mathcal{T}_{\tilde{\gamma},\tilde{\delta}}^A,S\big\rangle_{\mathcal{B}_2(L^2)}\big)_{(\tilde{\gamma},\tilde{\delta})\in\Z^{4d}}$ is a member of $\ell^2(\Z^{4d})$, and Theorem \ref{thm:matrix_rep} tells us that the matrix
    \begin{equation*}
        \mathbb{M}\coloneq\left(\sum_{(\tilde{\alpha},\tilde{\beta})\in\Z^{4d}}\overline{\mathbb{M}_{(\tilde{\alpha},\tilde{\beta}),(\tilde{\gamma},\tilde{\delta})}^A[ \mathfrak{Op}^A(\Phi)]}\mathbb{M}_{(\tilde{\alpha},\tilde{\beta}),(\tilde{\mu},\tilde{\nu})}^A[ \mathfrak{Op}^A(\Phi)]\right)_{(\tilde{\gamma},\tilde{\delta}),(\tilde{\mu},\tilde{\nu})\in\Z^{4d}}
    \end{equation*}
    satisfies the Schur test \cite{Hedenmalm2000} on $\ell^2(\Z^{4d})$, so
    \begin{equation*}
        \begin{aligned}
            \Vert\mathfrak{Op}^A(\Phi)S\Vert_{\mathcal{B}_2(L^2)}^2&\leq C\Vert\Phi\Vert_{S_0(M),k}^2\big\Vert\big(\big\langle \mathcal{T}_{\tilde{\gamma},\tilde{\delta}}^A,S\big\rangle_{\mathcal{B}_2(L^2)}\big)_{(\tilde{\gamma},\tilde{\delta})\in\Z^{4d}} \big\Vert_{\ell^2}^2\\
            &=C\Vert\Phi\Vert_{S_0(M),k}^2\Vert S \Vert_{\mathcal{B}_2(L^2)}^2
        \end{aligned}
    \end{equation*}
    for some $C>0,k\in\N_0$ from \eqref{eq:matrix_specific}.
    
    \paragraph{Proof of 2.} The assumption on $M$ implies that $M\in L^\infty(\R^{4d})$, and using Lemma \ref{lem:convergence_l2} twice we get an expansion of the kind \eqref{eq:supermatrix} but for Hilbert-Schmidt operators $S$ with convergence in $\mathcal{B}_2(L^2(\R^d))$-norm. This allows us to decompose $\mathfrak{Op}^A(\Phi)$ into the composition of three bounded operators, as done in the proof of \cite[Theorem 4.2]{Thorn2026}:
    \begin{equation*}
        \begin{aligned}
            \mathcal{B}(L^2(\R^d))\ni S&\mapsto \big(\big\langle \mathcal{T}_{\tilde{\gamma},\tilde{\delta}}^A,S\big\rangle_{L^2}\big)_{(\tilde{\gamma},\tilde{\delta})\in\Z^{4d}}\in\ell^2(\Z^{4d})\\
            \ell^2(\Z^{4d})\ni(c_{\tilde{\gamma},\tilde{\delta}})_{(\tilde{\gamma},\tilde{\delta})\in\Z^{4d}}&\mapsto\left(\sum_{(\tilde{\gamma},\tilde{\delta})\in\Z^{4d}}\mathbb{M}_{(\tilde{\alpha},\tilde{\beta}),(\tilde{\gamma},\tilde{\delta})}^A[\mathfrak{Op}_t^A(\Phi)]c_{\tilde{\gamma},\tilde{\delta}}\right)_{(\tilde{\alpha},\tilde{\beta})\in\Z^{2d}}\in\ell^2(\Z^{4d})\\
            \ell^2(\Z^{4d})\ni(c_{\tilde{\alpha},\tilde{\beta}})_{(\tilde{\alpha},\tilde{\beta})\in\Z^{4d}}&\mapsto\sum_{(\tilde{\alpha},\tilde{\beta})\in\Z^{4d}}c_{\tilde{\alpha},\tilde{\beta}}\mathcal{T}_{\tilde{\alpha},\tilde{\beta}}^A\in \mathcal{B}(L^2(\R^d))
        \end{aligned}
    \end{equation*}
    We prove that the map in the middle is compact, whence it follows that $\mathfrak{Op}^A(\Phi)$ is compact. This map can again be decomposed into the following two operators on $\ell^2(\Z^{4d})$:
    \begin{equation*}
        \begin{aligned}
            (c_{\tilde{\gamma},\tilde{\delta}})_{(\tilde{\gamma},\tilde{\delta})\in\Z^{4d}}&\mapsto\left(M(\tilde{\alpha},\tilde{\beta})^{-1}\sum_{(\tilde{\gamma},\tilde{\delta})\in\Z^{4d}}\mathbb{M}_{(\tilde{\alpha},\tilde{\beta}),(\tilde{\gamma},\tilde{\delta})}^A[\mathfrak{Op}_t^A(\Phi)]c_{\tilde{\gamma},\tilde{\delta}}\right)_{(\tilde{\alpha},\tilde{\beta})\in\Z^{2d}}\\
            (c_{\tilde{\alpha},\tilde{\beta}})_{(\tilde{\alpha},\tilde{\beta})\in\Z^{4d}}&\mapsto (M(\tilde{\alpha},\tilde{\beta})c_{\tilde{\alpha},\tilde{\beta}})
        \end{aligned}
    \end{equation*}
    Using Theorem \ref{thm:matrix_rep}, Peetre's inequality \eqref{eq:peetre}, and the Schur test \cite{Hedenmalm2000}, we see that the first operator is bounded. The second is a multiplication operator in $\ell^2(\Z^{4d})$ with multiplier converging to zero at infinity, hence compact.
    
    \paragraph{Proof of 3.} We first note that $M\in L^p(\R^{4d})$ with $M$ being a tempered weight implies the assumption in \textit{2.}, so $\mathfrak{Op}^A(\Phi)$ is compact on $\mathcal{B}_2(L^2(\R^d))$.
    
    Secondly, we reduce the statement to $p\in(0,2]$. Suppose $p\in(2,\infty)$. For a compact operator $S$ on a Hilbert space, being in the $p$-Schatten-class is equivalent to $S^*S$ being in the $\frac{p}{2}$-Schatten-class, and the following holds:
    \begin{equation*}
        \Vert S\Vert_{\mathcal{B}_p}=\Vert S^*S\Vert_{\mathcal{B}_{\frac{p}{2}}}^\frac{1}{2}
    \end{equation*}
    Considering Proposition \ref{prop:calc_super_moyal} on the magnetic super Moyal algebra and \cite[Proposition IV.7]{LeeLein2022} on $\mathcal{B}_2(L^2(\R^d))$-adjoints, if we could prove the case $p\in(0,2]$, then using the above fact repeatedly would also prove the case $p\in(2,\infty)$.

    Now to prove the case $p\in(0,2]$ we use \cite[Theorem B]{Zhu2013}: Essentially, $\mathfrak{Op}^A(\Phi)$ is in the $p$-Schatten-class if
    \begin{equation*}
        \sum_{\tilde{\gamma},\tilde{\delta}\in\Z^{2d}}\Vert\mathfrak{Op}^A(\Phi)\mathcal{T}_{\tilde{\gamma},\tilde{\delta}}^A\Vert_{\mathcal{B}_2(L^2)}^p<\infty
    \end{equation*}
    and
    \begin{equation*}
        \Vert\mathfrak{Op}^A(\Phi)\Vert_{\mathcal{B}_p(\mathcal{B}_2(L^2))}^p\leq\sum_{\tilde{\gamma},\tilde{\delta}\in\Z^{2d}}\Vert\mathfrak{Op}^A(\Phi)\mathcal{T}_{\tilde{\gamma},\tilde{\delta}}^A\Vert_{\mathcal{B}_2(L^2)}^p.
    \end{equation*}
    We estimate using Theorem \ref{thm:matrix_rep}:
    \begin{equation*}
        \begin{aligned}
            \sum_{\tilde{\gamma},\tilde{\delta}\in\Z^{2d}}\Vert\mathfrak{Op}^A(\Phi)\mathcal{T}_{\tilde{\gamma},\tilde{\delta}}^A\Vert_{\mathcal{B}_2(L^2)}^p&=\sum_{\tilde{\gamma},\tilde{\delta}\in\Z^{2d}}\left(\sum_{\tilde{\alpha},\tilde{\beta}\in\Z^{2d}}\bigg|\mathbb{M}_{(\tilde{\alpha},\tilde{\beta}),(\tilde{\gamma},\tilde{\delta})}^A[\mathfrak{Op}^A(\Phi)]\bigg|^2\right)^\frac{p}{2}\\
            &=\sum_{\tilde{\gamma},\tilde{\delta}\in\Z^{2d}}\left(\sum_{\tilde{\alpha},\tilde{\beta}\in\Z^{2d}}\bigg|\mathbb{M}_{(\tilde{\alpha},\tilde{\beta}),(\tilde{\gamma},\tilde{\delta})}^A[\mathfrak{Op}^A(\Phi)]\bigg|^{p\frac{2}{p}}\right)^\frac{p}{2}\\
            &\leq\sum_{\tilde{\alpha},\tilde{\beta},\tilde{\gamma},\tilde{\delta}\in\Z^{2d}}\bigg|\mathbb{M}_{(\tilde{\alpha},\tilde{\beta}),(\tilde{\gamma},\tilde{\delta})}^A[\mathfrak{Op}^A(\Phi)]\bigg|^p\leq C\Vert\Phi\Vert_{S_0(M),k}
        \end{aligned}
    \end{equation*}
    for some $C>0,k\in\N_0$, where in the second to last step we used that $2/p>1$. This shows that $\mathfrak{Op}^A(\Phi)\in\mathcal{B}_p\big(\mathcal{B}_2(L^2(\R^d))\big)$ and the associated continuity assertion.
\end{proof}

\begin{remark}
    It was strictly only necessary to have the assumptions hold for $M$ on some lattices, but like our Remark \ref{rem:app1}, this is equivalent to the given assumptions up to measurability, see Lemma \ref{lem:app_lebesgue} and \ref{lem:app_tozero}.
\end{remark}

\section{A Beals-type Super Commutator Criterion}\label{sec:commutator}

The magnetic Beals criterion was initially considered in \cite[Theorem 1.1]{Iftimie2010}, while in  \cite{CorneanHelfferPurice2018,CorneanHelfferPurice2024} the subject was revisited using frames. Here we extend those  results to super operators.

\begin{theorem}\label{thm:beals}
    Let $s_L$, $s_R$, $s_L'$, $s_R'\in\R$ and consider a bounded super operator $$\mathfrak{T}\colon\mathcal{B}(\mathscr{S}'(\R^d),\mathscr{S}(\R^d))\rightarrow\mathcal{B}(\mathscr{S}(\R^d),\mathscr{S}'(\R^d)).$$ 
    The following two statements are equivalent:
    \begin{enumerate}[label={\rm(\roman*)}, ref={\rm(\roman*)}]
        \item\label{i_beals} All commutators of the kind
        \begin{equation}\label{eq:commutators}
            [\mathfrak{S}_n,[\cdots [\mathfrak{S}_2,[\mathfrak{S}_1,\mathfrak{T}]]\cdots]],
        \end{equation}
        with $n\in\N_0$ and $\mathfrak{S}_j\in\{\mathrm{Id}\odot Q_1,Q_1\odot\mathrm{Id},\dots,Q_d\odot\mathrm{Id},\mathrm{Id}\odot P_1^A,P_1^A\odot\mathrm{Id},\dots,P_d^A\odot\mathrm{Id}\}$ for $j=1,\dots,n$, are bounded from $\mathcal{B}_2\big(H_A^{s_R}(\R^d),H_A^{s_L}(\R^d)\big)$ into $\mathcal{B}_2\big(H_A^{s_R'}(\R^d),H_A^{s_L'}(\R^d)\big)$.
        
        \item\label{ii_beals} $\mathfrak{T}$ has a super symbol in $S_0(m_0^{s_L-s_L'}\otimes m_0^{s_R'-s_R})$.
    \end{enumerate}
\end{theorem}

\begin{proof}[Proof that \ref{i_beals} implies \ref{ii_beals}]
    We start by showing that the assumption about the commutators \eqref{eq:commutators} implies that the super symbol of $\mathfrak{T}$ belongs to $S_0(m_0^{s_L-s_L'}\otimes m_0^{s_R'-s_R})$. 

    Suppose $\mathfrak{T}$ is a super operator for which all the commutators of the form \eqref{eq:commutators} are bounded from $\mathcal{B}_2\big(H_A^{s_R}(\R^d),H_A^{s_L}(\R^d)\big)$ into $\mathcal{B}_2\big(H_A^{s_R'}(\R^d),H_A^{s_L'}(\R^d)\big)$. We will show that this implies that the matrix elements of $\mathfrak{T}$ in the frame $(\mathcal{T}_{\tilde{\alpha},\tilde{\beta}}^A)_{\tilde{\alpha},\tilde{\beta}\in\Z^{2d}}$ satisfy the condition in \eqref{eq:matrix} with $M= m_0^{s_L-s_L'}\otimes m_0^{s_R'-s_R}$, which then implies that the super symbol of $\mathfrak{T}$ must be in $S_0(m_0^{s_L-s_L'}\otimes m_0^{s_R'-s_R})$ by Theorem \ref{thm:matrix_rep} \ref{ii_matrix_rep}.

    Let us first deal with the arbitrary decay in the diagonal of $\mathbb{M}^A[\mathfrak{T}]$. Following the strategy of \cite{CorneanHelfferPurice2018}, we see that
    \begin{equation}\label{eq:beals_position}
        \begin{aligned}
            (\alpha_j-\gamma_j)\mathbb{M}_{(\tilde{\alpha},\tilde{\beta}),(\tilde{\gamma},\tilde{\delta})}^A[\mathfrak{T}]&=\Big\langle[(Q_j\odot \mathrm{Id}),\mathfrak{T}]\mathcal{T}_{\tilde{\gamma},\tilde{\delta}}^A,\overline{\mathcal{T}_{\tilde{\alpha},\tilde{\beta}}^A}\Big\rangle_{\mathcal{B}',\mathcal{B}}\\
            &\quad-\Big\langle\mathfrak{T}\mathcal{T}_{\tilde{\gamma},\tilde{\delta}}^A,\big((Q_j\odot \mathrm{Id})-\alpha_j\big)\overline{\mathcal{T}_{\tilde{\alpha},\tilde{\beta}}^A}\Big\rangle_{\mathcal{B}',\mathcal{B}}\\
            &\quad+\Big\langle\mathfrak{T}\big((Q_j\odot \mathrm{Id})-\gamma_j\big)\mathcal{T}_{\tilde{\gamma},\tilde{\delta}}^A,\overline{\mathcal{T}_{\tilde{\alpha},\tilde{\beta}}^A}\Big\rangle_{\mathcal{B}',\mathcal{B}}
        \end{aligned}
    \end{equation}
    and
    \begin{equation}\label{eq:beals_momentum}
        \begin{aligned}
            (\alpha_j'-\gamma_j')\mathbb{M}_{(\tilde{\alpha},\tilde{\beta}),(\tilde{\gamma},\tilde{\delta})}^A[\mathfrak{T}]&=\Big\langle[(P_j^A\odot \mathrm{Id}),\mathfrak{T}]\mathcal{T}_{\tilde{\gamma},\tilde{\delta}}^A,\overline{\mathcal{T}_{\tilde{\alpha},\tilde{\beta}}^A}\Big\rangle_{\mathcal{B}',\mathcal{B}}\\
            &\quad-\Big\langle\mathfrak{T}\mathcal{T}_{\tilde{\gamma},\tilde{\delta}}^A,\big((P_j^A\odot \mathrm{Id})-\alpha_j'\big)\overline{\mathcal{T}_{\tilde{\alpha},\tilde{\beta}}^A}\Big\rangle_{\mathcal{B}',\mathcal{B}}\\
            &\quad+\Big\langle\mathfrak{T}\big((P_j^A\odot \mathrm{Id})-\gamma_j'\big)\mathcal{T}_{\tilde{\gamma},\tilde{\delta}}^A,\overline{\mathcal{T}_{\tilde{\alpha},\tilde{\beta}}^A}\Big\rangle_{\mathcal{B}',\mathcal{B}}
        \end{aligned}
    \end{equation}
    with $j=1,\dots,d$ and similar decompositions holding for $\beta_j-\delta_j$ and $\beta_j'-\delta_j'$ using $\mathrm{Id}\odot Q_j$ and $\mathrm{Id}\odot P_j^A$. All three terms of the right hand side have the same structure: The super operators $\mathfrak{T}$, $[Q_j\odot\mathrm{Id},\mathfrak{T}]$, and $[P_j^A\odot\mathrm{Id},\mathfrak{T}]$ are by assumption bounded between $\mathcal{B}_2\big(H_A^{s_R}(\R^d),H_A^{s_L}(\R^d)\big)$ and $\mathcal{B}_2\big(H_A^{s_R'}(\R^d),H_A^{s_L'}(\R^d)\big)$, while $\mathcal{T}_{\tilde{\alpha},\tilde{\beta}}^A$,
    \begin{equation*}
        \big((Q_j\odot \mathrm{Id})-\alpha_j\big)\mathcal{T}_{\tilde{\alpha},\tilde{\beta}}^A=\mathfrak{int}\Big(\big[\big(Q_j-\alpha_j\big)\mathcal{G}_{\tilde{\alpha}}^A\big]\otimes\overline{\mathcal{G}_{\tilde{\beta}}^A}\Big),
    \end{equation*}
    and
    \begin{equation*}
        \big((P_j^A\odot \mathrm{Id})-\alpha_j'\big)\mathcal{T}_{\tilde{\alpha},\tilde{\beta}}^A=\mathfrak{int}\Big(\big[\big(P_j^A-\alpha_j'\big)\mathcal{G}_{\tilde{\alpha}}^A\big]\otimes\overline{\mathcal{G}_{\tilde{\beta}}^A}\Big)
    \end{equation*}
    are rank one integral operators. Here
    \begin{equation*}
        \big[\big(Q_j-\alpha_j\big)\mathcal{G}_{\tilde{\alpha}}^A\big](x)=(2\pi)^{-\frac{d}{2}}e^{i\varphi(x,\alpha)}(x_j-\alpha_j)\mathfrak{g}(x-\alpha)e^{i\alpha'\cdot(x-\alpha)}
    \end{equation*}
    and using \eqref{eq:magpotential}
    \begin{equation*}
        \big[\big(P_j^A-\alpha_j'\big)\mathcal{G}_{\tilde{\alpha}}^A\big](x)=(2\pi)^{-\frac{d}{2}}e^{i\varphi(x,\alpha)}\big(-A_j(x;\alpha)\mathfrak{g}(x-\alpha)-i\partial_j\mathfrak{g}(x-\alpha)\big)e^{i\alpha'\cdot(x-\alpha)},
    \end{equation*}
    where $(x_j-\alpha_j)\mathfrak{g}(x-\alpha)$ and $-A_j(x;\alpha)\mathfrak{g}(x-\alpha)-i\partial_j\mathfrak{g}(x-\alpha)$ are uniformly bounded functions in $x$ and $\alpha$.

    By repeatedly using \eqref{eq:beals_position} and \eqref{eq:beals_momentum} in combination with the decompositions for $\beta_j-\delta_j$ and $\beta_j'-\delta_j'$ we find that
    \begin{equation*}
        \langle(\tilde{\alpha},\tilde{\beta})-(\tilde{\gamma},\tilde{\delta})\rangle^{2k}\mathbb{M}_{(\tilde{\alpha},\tilde{\beta}),(\tilde{\gamma},\tilde{\delta})}^A[\mathfrak{T}]
    \end{equation*}
    with $k\in\N_0$ equals a finite sum with every term having the structure
    \begin{equation}\label{eq:beals_form_of_term}
        \Big\langle\tilde{\mathfrak{T}}\, \mathfrak{int}(G_{\tilde{\gamma},1}\otimes G_{\tilde{\delta},2}),\mathfrak{int}(G_{\tilde{\alpha},3}\otimes G_{\tilde{\beta},4})\Big\rangle_{\mathcal{B}',\mathcal{B}},
    \end{equation}
    where $\tilde{\mathfrak{T}}$ is bounded from $\mathcal{B}_2\big(H_A^{s_R}(\R^d),H_A^{s_L}(\R^d)\big)$ into $\mathcal{B}_2\big(H_A^{s_R'}(\R^d),H_A^{s_L'}(\R^d)\big)$, and
    \begin{equation*}
        G_{\tilde{\alpha},3}(x)=e^{i\varphi(x,\alpha)}g_3(x,\alpha)e^{i\alpha'\cdot(x-\alpha)},
    \end{equation*}
    where $g_3$ is smooth and uniformly bounded, with $G_{\tilde{\gamma},1}, G_{\tilde{\delta},2},G_{\tilde{\beta},4}$ being of the same form. Each of the terms \eqref{eq:beals_form_of_term} can be estimated as follows:
    \begin{align*}
        \Big|\Big\langle&\tilde{\mathfrak{T}}\, \mathfrak{int}(G_{\tilde{\gamma},1}\otimes G_{\tilde{\delta},2}),\mathfrak{int}(G_{\tilde{\alpha},3}\otimes G_{\tilde{\beta},4})\Big\rangle_{\mathcal{B}',\mathcal{B}}\Big|\\
        &=\Big|\Big\langle\mathfrak{Op}^A(\chi_{s_L'}\otimes\chi_{-s_R'})\, \tilde{\mathfrak{T}}\, \mathfrak{Op}^A(\chi_{-s_L}\otimes\chi_{s_R})\mathfrak{Op}^A(\chi_{s_L}\otimes\chi_{-s_R})\mathfrak{int}(G_{\tilde{\gamma},1}\otimes G_{\tilde{\delta},2}),\\
    &\qquad\qquad\qquad\mathfrak{Op}^A(\chi_{-s_L'}\otimes\chi_{s_R'})\mathfrak{int}(G_{\tilde{\alpha},3}\otimes G_{\tilde{\beta},4})\Big\rangle_{\mathcal{B}',\mathcal{B}}\Big|\\
        &\leq C_1 \Vert\tilde{\mathfrak{T}} \Vert_{\mathcal{B}} \big\Vert\mathfrak{Op}^A(\chi_{s_L}\otimes\chi_{-s_R})\mathfrak{int}(G_{\tilde{\gamma},1}\otimes G_{\tilde{\delta},2})\big\Vert_{\mathcal{B}_2(L^2)}\\
&\quad\times\big\Vert\mathfrak{Op}^A(\chi_{-s_L'}\otimes\chi_{s_R'})\mathfrak{int}(G_{\tilde{\alpha},3}\otimes G_{\tilde{\beta},4})\big\Vert_{\mathcal{B}_2(L^2)}\\
        &\leq C_1\Vert\tilde{\mathfrak{T}}\Vert_{\mathcal{B}} \Vert\mathfrak{op}^A(\chi_{s_L})G_{\tilde{\gamma},1}\Vert_{L^2}\Vert\mathfrak{op}^A(\chi_{-s_R})G_{\tilde{\delta},2}\Vert_{L^2}\Vert\mathfrak{op}^A(\chi_{-s_L'})G_{\tilde{\alpha},3}\Vert_{L^2}\Vert\mathfrak{op}^A(\chi_{s_R'})G_{\tilde{\beta},4}\Vert_{L^2}
    \end{align*}
    with $C_1>0$. Using the same strategy as the proof of \cite[Theorem 3.4 1]{Thorn2026} (when estimating $\langle k^A\phi,\overline{\mathcal{G}_{\tilde{\alpha}}^A}\otimes \mathcal{G}_{\tilde{\beta}}^A\rangle_{\mathscr{S}',\mathscr{S}}$) we get:
    \begin{align*}
        \Vert\mathfrak{op}^A(\chi_{s_L})G_{\tilde{\gamma},1}\Vert_{L^2}&=\sqrt{\langle \mathfrak{op}^A(\chi_{s_L}\mathbin{\star^B}\chi_{s_L})G_{\tilde{\gamma},1},\overline{G_{\tilde{\gamma},1}}\rangle_{\mathscr{S}',\mathscr{S}}}\\
        &=\sqrt{\langle k^A(\chi_{s_L}\mathbin{\star^B}\chi_{s_L}),\overline{G_{\tilde{\gamma},1}}\otimes G_{\tilde{\gamma},1}\rangle_{\mathscr{S}',\mathscr{S}}}\leq C_2m_0^{s_L}(\tilde{\gamma})
    \end{align*}    
    where $C_2>0$.  Similar estimates can be computed for $\Vert\mathfrak{op}^A(\chi_{-s_R})G_{\tilde{\delta},2}\Vert_{L^2}$ and the other factors.

    Collecting all of this and using Peetre's inequality and \eqref{eq:peetre}, we see that for a chosen $k\in\N_0$:
    \begin{align*}
        \Big|\mathbb{M}_{(\tilde{\alpha},\tilde{\beta}),(\tilde{\gamma},\tilde{\delta})}^A[\mathfrak{T}]\Big|&=\langle(\tilde{\alpha},\tilde{\beta})-(\tilde{\gamma},\tilde{\delta})\rangle^{-2k}\langle(\tilde{\alpha},\tilde{\beta})-(\tilde{\gamma},\tilde{\delta})\rangle^{2k}\Big|\mathbb{M}_{(\tilde{\alpha},\tilde{\beta}),(\tilde{\gamma},\tilde{\delta})}^A[\mathfrak{T}]\Big|\\
        &\leq C_3\langle(\tilde{\alpha},\tilde{\beta})-(\tilde{\gamma},\tilde{\delta})\rangle^{-2k}m_0^{s_L}(\tilde{\gamma})m_0^{-s_R}(\tilde{\delta})m_0^{-s_L'}(\tilde{\alpha})m_0^{s_R'}(\tilde{\beta})\\
        &\leq C_4 \langle(\tilde{\alpha},\tilde{\beta})-(\tilde{\gamma},\tilde{\delta})\rangle^{-2k+a}m_0^{s_L-s_L'}\bigg(\frac{\tilde{\alpha}+\tilde{\gamma}}{2}\bigg)m_0^{s_R'-s_R}\bigg(\frac{\tilde{\beta}+\tilde{\delta}}{2}\bigg)
    \end{align*}
    with $a,C_3,C_4>0$ depending on $s_L,s_R,s_L',s_R'$ through Peetre's inequality, and $C_3,C_4$ also depending on $k$.
\end{proof}

To show that \ref{ii_beals} implies \ref{i_beals} we first prove the following:

\begin{proposition}\label{prop:not_quite_beals}
    For any tempered weight $M$ on $\R^{4d}$ and any super symbol $\Phi\in S_0(M)$, all commutators of the kind \eqref{eq:commutators} with $\mathfrak{T}=\mathfrak{Op}^A(\Phi)$ have super symbols in $S_0(M)$.
\end{proposition}
\begin{proof}
    The proof proceeds by induction on the number of commutators $n$, so let us begin by showing that a single commutator $[\mathfrak{S}_1,\mathfrak{Op}^A(\Phi)]$ dequantizes to a super symbol in $S_0(M)$. Since the strategy is the same for every choice of $\mathfrak{S}_1$ we suppose that $\mathfrak{S}_1=\mathrm{Id}\odot Q_1$.

    By Theorem \ref{thm:matrix_rep} \ref{ii_matrix_rep}, it is enough to show that the matrix elements of $[\mathfrak{S}_1,\mathfrak{Op}^A(\Phi)]$ satisfy \eqref{eq:matrix}. These are given by:
    \begin{align*}
        \mathbb{M}_{(\tilde{\alpha},\tilde{\beta}),(\tilde{\gamma},\tilde{\delta})}^A\big[[\mathfrak{S}_1,\mathfrak{Op}^A(\Phi)]\big]
        &=\Big\langle\mathfrak{Op}^A(\Phi)\mathcal{T}_{\tilde{\gamma},\tilde{\delta}}^A,(\mathfrak{S}_1-\beta_1)\overline{\mathcal{T}_{\tilde{\alpha},\tilde{\beta}}^A}\Big\rangle_{\mathcal{B}',\mathcal{B}}\\
        &\quad-\Big\langle\mathfrak{Op}^A(\Phi)(\mathfrak{S}_1-\delta_1)\mathcal{T}_{\tilde{\gamma},\tilde{\delta}}^A,\overline{\mathcal{T}_{\tilde{\alpha},\tilde{\beta}}^A}\Big\rangle_{\mathcal{B}',\mathcal{B}}\\
        &\quad+(\delta_1-\beta_1)\Big\langle\mathfrak{Op}^A(\Phi)\mathcal{T}_{\tilde{\gamma},\tilde{\delta}}^A,\overline{\mathcal{T}_{\tilde{\alpha},\tilde{\beta}}^A}\Big\rangle_{\mathcal{B}',\mathcal{B}}\\
        &=\big\langle K^A\Phi,\mathcal{G}_{\tilde{\gamma}}^A\otimes\overline{\mathcal{G}_{\tilde{\delta}}^A}\otimes\overline{\mathcal{G}_{\tilde{\alpha}}^A}\otimes(Q_1-\beta_1)\mathcal{G}_{\tilde{\beta}}^A)\big\rangle_{\mathscr{S}',\mathscr{S}}\\
        &\quad-\big\langle K^A\Phi,\mathcal{G}_{\tilde{\gamma}}^A\otimes\overline{(Q_1-\delta_1)\mathcal{G}_{\tilde{\delta}}^A}\otimes\overline{\mathcal{G}_{\tilde{\alpha}}^A}\otimes\mathcal{G}_{\tilde{\beta}}^A\big\rangle_{\mathscr{S}',\mathscr{S}}\\
        &\quad+(\delta_1-\beta_1)\big\langle K^A\Phi,\mathcal{G}_{\tilde{\gamma}}^A\otimes\overline{\mathcal{G}_{\tilde{\delta}}^A}\otimes\overline{\mathcal{G}_{\tilde{\alpha}}^A}\otimes\mathcal{G}_{\tilde{\beta}}^A\big\rangle_{\mathscr{S}',\mathscr{S}}
    \end{align*}
    Each of the three terms can be estimated in the exact same way as the matrix elements  $\mathbb{M}_{(\tilde{\alpha},\tilde{\beta}),(\tilde{\gamma},\tilde{\delta})}^A[\mathfrak{Op}^A(\Phi)]$, i.e. using the outline in the proof of Theorem \ref{thm:matrix_rep} \ref{i_matrix_rep}.

    For the induction step we assume that for all choices of $\mathfrak{S}_1,\dots,\mathfrak{S}_n$ we have 
    \begin{equation*}
        [\mathfrak{S}_n,[\cdots[\mathfrak{S}_2,[\mathfrak{S}_1,\mathfrak{Op}^A(\Phi)]]\cdots]=\mathfrak{Op}^A(\tilde{\Phi})
    \end{equation*}
    with $\tilde{\Phi}\in S_0(M)$. Thus for some $\mathfrak{S}_{n+1}\in\{\mathrm{Id}\odot Q_1,\dots,P_d^A\odot\mathrm{Id}\}$, we consider the super operator:
    \begin{equation*}
        [\mathfrak{S}_{n+1},[\mathfrak{S}_n,[\cdots[\mathfrak{S}_2,[\mathfrak{S}_1,\mathfrak{Op}^A(\Phi)]]\cdots]]=[\mathfrak{S}_{n+1},\mathfrak{Op}^A(\tilde{\Phi})]
    \end{equation*}
    Now to show that this super operator has a super symbol in $S_0(M)$ one uses the exact same method as in the basis step. This also concludes the proof of Theorem \ref{thm:beals}. 
    \end{proof}

\begin{proof}[Proof that \ref{ii_beals} implies \ref{i_beals}]
    This now follows from Proposition \ref{prop:not_quite_beals} and Theorem \ref{thm:caldvail}.
\end{proof}

\begin{remark}
    The analog of Proposition \ref{prop:not_quite_beals} for magnetic pseudo-differential operators has, as far as we know, only been proven for the Hörmander class $s_0(1)=s_{0,0}^0$. With the above techniques and the framework of \cite{Thorn2026} one could prove a similar result holding for all the classes $s_0(m)$:
    
    \begin{proposition}
        For any tempered weight $m$ on $\R^{2d}$ and any symbol $\phi\in S_0(M)$, all commutators of the kind
        \begin{equation*}
            [S_n,[\cdots [S_2,[S_1,\mathfrak{op}^A(\phi)]]\cdots]],
        \end{equation*}
        with $n\in\N_0$ and $S_j\in\{Q_1,\dots,Q_d,P_1^A,\dots,P_d^A\}$, have symbols in $s_0(m)$.
    \end{proposition}

    An interesting question is whether such a strengthening of results exists in the converse statement, e.g. if all the commutators \eqref{eq:commutators} of a super operator $\mathfrak{T}$ are compact on $\mathcal{B}_2(L^2(\R^d))$ is it then true that the super symbol of $\mathfrak{T}$ and all of its derivatives decay to zero at infinity? If the answer is yes, does the super symbol belong to the class $S_0(M)$ were perhaps $M(\tilde{\alpha},\tilde{\beta})\approx\big|\big\langle\mathfrak{T}\mathcal{T}_{\tilde{\alpha},\tilde{\beta}}^A,\overline{\mathcal{T}_{\tilde{\alpha},\tilde{\beta}}^A}\big\rangle_{\mathcal{B}',\mathcal{B}}\big|$?
    
    Of course a similar question could be stated for symbols of operators with certain commutator properties.
\end{remark}

\section{Complete Positivity and Trace Preservation}\label{sec:positivity}

We now move on to our last topic of complete positivity \cite{Stinespring1955} and trace preservation, where our main result is heavily inspired by the papers of Hellwig and Kraus \cite{Kraus1970,Kraus1971}. These papers study super operators that are sums of the ``simple" super operators $S_n\odot S_n^*$, which would correspond to super symbols that are sums of the ``simple" super symbols $\phi_n\otimes\overline{\phi_n}$ in the magnetic pseudo-differential super calculus. Their main assumption on the operators $(S_n)_{n\in\N_0}$ is that $\sum_{n=0}^\infty S_n^*S_n$ converges to the identity in some sense. Translating this into conditions on symbols we get the following:

\begin{theorem}\label{thm:cptp}
    Let $(\phi_n)_{n\in\N_0}$ be a sequence of symbols in $s_0(1)$ for which the partial sums of\footnote{Recall that $\star^B$ denotes the magnetic Moyal product \cite{Mantoiu2004,Iftimie2007}.} $\sum_{n=0}^\infty \overline{\phi_n}\mathbin{\star^B}\phi_n$ are uniformly bounded in $s_0(1)$ and $\sum_{n=0}^\infty \overline{\phi_n}\mathbin{\star^B}\phi_n=1$ with pointwise convergence.

    Then $\sum_{n=0}^\infty\phi_n\otimes\overline{\phi_n}$ converges to a tempered distribution $\Phi$ in the weak*-topology of $\mathscr{S}'(\R^{4d})$ and the super operator $\mathfrak{Op}^A(\Phi)$ is bounded, completely positive, and trace preserving on $\mathcal{B}_1(L^2(\R^d))$.
\end{theorem}

\begin{proof}
    Our strategy is to show that the quantizations of $\sum_{n=0}^N\phi_n\otimes\overline{\phi_n}$ converge strongly to a super operator $\mathfrak{T}$ on $\mathcal{B}_1(L^2(\R^d))$, and that $\mathfrak{T}$ is bounded, completely positive, and trace preserving. We then prove that the sum $\sum_{n=0}^\infty\phi_n\otimes\overline{\phi_n}$ converges and that the limit equals the super symbol of $\mathfrak{T}$. But before we are able to do any of this, we need to analyze the quantizations of $\sum_{n=0}^N \overline{\phi_n}\mathbin{\star^B}\phi_n$.

    Let $\psi_N\coloneq\sum_{n=0}^N \overline{\phi_n}\mathbin{\star^B}\phi_n$ for $N\in\N_0$. We show that $\mathfrak{op}^A(\psi_N)\rightarrow\mathrm{Id}$ ultraweakly in $L^2(\R^d)$ as $N\rightarrow\infty$. From the magnetic Calderón-Vaillancourt theorem \cite[Theorem 3.1]{Iftimie2007} we see that $(\mathfrak{op}^A(\psi_N))_{N\in\N_0}$ are bounded operators on $L^2(\R^d)$ with uniformly bounded operator norm. Furthermore, the sequence $(\psi_N)_{N\in\N_0}$ is bounded in $s_0(1)$ and converge pointwise to $1$, so $\langle\mathfrak{op}^A(\psi_N)f,g\rangle_{\mathscr{S}',\mathscr{S}}\rightarrow\langle f,g\rangle_{\mathscr{S}',\mathscr{S}}$ when $N\rightarrow\infty$ for $f,g\in\mathscr{S}(\R^d)$. If $f,g\in L^2(\R^d)$ we can find $f_\varepsilon,g_\varepsilon\in\mathscr{S}(\R^d)$ for which $\Vert f-f_\varepsilon\Vert_{L^2},\Vert g-g_\varepsilon\Vert_{L^2}<\varepsilon<1$, and so:
    \begin{align*}
        |\langle g,(\mathfrak{op}^A(\psi_N)&-\mathrm{Id})f\rangle_{L^2}|\\
        &\leq|\langle g,(\mathfrak{op}^A(\psi_N)-\mathrm{Id})(f-f_\varepsilon)\rangle_{L^2}|+|\langle g-g_\varepsilon,(\mathfrak{op}^A(\psi_N)-\mathrm{Id})f_\varepsilon\rangle_{L^2}|\\
        &\quad+|\langle g_\varepsilon,(\mathfrak{op}^A(\psi_N)-\mathrm{Id})f_\varepsilon\rangle_{L^2}|\\
        &\leq C\varepsilon+|\langle \overline{g_\varepsilon},(\mathfrak{op}^A(\psi_N)-\mathrm{Id})f_\varepsilon\rangle_{\mathscr{S}',\mathscr{S}}|
    \end{align*}
    for some $C>0$ not dependent on $f_\varepsilon,g_\varepsilon$. This implies that:
    \begin{align*}
        \limsup_{N\rightarrow\infty}|\langle g,(\mathfrak{op}^A(\psi_N)-\mathrm{Id})f\rangle_{L^2}|\leq C\varepsilon
    \end{align*}
    Letting $\varepsilon\rightarrow0$ we see that $\mathfrak{op}^A(\psi_N)$ converges weakly to $\mathrm{Id}$ in $L^2(\R^d)$ when $N\rightarrow\infty$. Now the fact that $(\mathfrak{op}^A(\psi_N))_{N\in\N_0}$ is an increasing sequence of non-negative operators together with the weak convergence implies that they converge ultraweakly to $\mathrm{Id}$, see \cite{Kraus1971} and the references therein.

    Next we consider the super quantizations of $\Phi_N\coloneq\sum_{n=0}^N \phi_n\otimes\overline{\phi_n}$ for $N\in\N_0$. Specifically, we show that $\big(\mathfrak{Op}^A(\Phi_N)\big)_{N\in\N_0}$ converge strongly on $\mathcal{B}_1(L^2(\R^d))$ to a bounded super operator. Note that $\mathfrak{Op}^A(\Phi_N)$ are bounded on $\mathcal{B}_1(L^2(\R^d))$ as we saw in the proof of Proposition \ref{prop:continuoustensor}.

    Let $S\in\mathcal{B}_1(L^2(\R^d))$. We write $S$ as a linear combination of positive trace-class operators $S=R_1-R_2+i(R_3-R_4)$ and see that for $N_1>N_2$:
    \begin{align*}
        \Big\Vert \mathfrak{Op}^A(\Phi_{N_1}-\Phi_{N_2})S\Big\Vert_{\mathcal{B}_1(L^2)}&\leq\sum_{j=1}^4\sum_{n=N_2+1}^{N_1}\tr\big(\mathfrak{op}^A(\phi_n)R_j\mathfrak{op}^A(\overline{\phi_n})\big)\\
        &=\sum_{j=1}^4\tr(\mathfrak{op}^A(\psi_{N_1}-\psi_{N_2})R_j)
    \end{align*}
    Since $\mathfrak{op}^A(\psi_N)\rightarrow\mathrm{Id}$ ultraweakly as $N\rightarrow\infty$ we conclude that the right hand side of the above goes to zero as $N_1,N_2\rightarrow\infty$, see \cite{Kraus1971}. So $\big(\mathfrak{Op}^A(\Phi_N)\big)_{N\in\N_0}$ converge strongly to a super operator $\mathfrak{T}$ on trace-class operators. We prove that $\mathfrak{T}$ is bounded by showing that $\big(\mathfrak{Op}^A(\Phi_N)\big)_{N\in\N_0}$ are uniformly bounded: For $S\in\mathcal{B}_1(L^2(\R^d))$ we have
    \begin{align*}
        \big\Vert \mathfrak{Op}^A(\Phi_N)S\big\Vert_{\mathcal{B}_1(L^2)}&=\sup_{\Vert R\Vert_{\mathcal{B}(L^2)}=1}\big|\tr\big(R\, \mathfrak{Op}^A(\Phi_N)S\big)\big|\\
        &=\sup_{\Vert R\Vert_{\mathcal{B}(L^2)}=1}\Bigg|\tr\Bigg(\sum_{n=0}^N\mathfrak{op}^A(\phi_n)^*R\, \mathfrak{op}^A(\phi_n)S\Bigg)\Bigg|\\
        &\leq\Vert S\Vert_{\mathcal{B}_1(L^2)}\sup_{\Vert R\Vert_{\mathcal{B}(L^2)}=1}\Bigg\Vert\sum_{n=0}^N\mathfrak{op}^A(\phi_n)^* R\, \mathfrak{op}^A(\phi_n)\Bigg\Vert_{\mathcal{B}(L^2)}\\
        &\leq \Vert S\Vert_{\mathcal{B}_1(L^2)},
    \end{align*}
    where in the last step we used the fact that 
    $$\Bigg\Vert\sum_{n=0}^N\mathfrak{op}^A(\phi_n)^* R\, \mathfrak{op}^A(\phi_n)\Bigg\Vert_{\mathcal{B}(L^2)}\leq \Vert R\Vert_{\mathcal{B}(L^2)} =1$$
    as long as 
    $\mathfrak{op}^A(\psi_N)=\sum_{n=0}^N\mathfrak{op}^A(\phi_n)^*\mathfrak{op}^A(\phi_n)\leq\mathrm{Id}$. Hence $\big(\mathfrak{Op}^A(\Phi_N)\big)_{N\in\N_0}$ are uniformly bounded by $1$, implying that $\mathfrak{T}$ is a bounded super operator.

    Now, $\mathfrak{T}$ is trace preserving since
    \begin{equation*}
        \tr(\mathfrak{T}S)=\lim_{N\rightarrow\infty}\tr\Big(\mathfrak{Op}^A(\Phi_N)S\Big)=\lim_{N\rightarrow\infty}\tr(\mathfrak{op}^A(\psi_N)S)=\tr(S)
    \end{equation*}
    with the last identity following from the ultraweak convergence of $(\mathfrak{op}^A(\psi_N))_{N\in\N_0}$. It is also completely positive since for $k\in\N$,
    \begin{equation*}
        \big(\mathfrak{T}\otimes\mathrm{Id}_{\mathcal{B}(\C^k)}\big)S=\sum_{n=0}^\infty\big(\mathfrak{op}^A(\phi_n)\otimes\mathrm{Id}_{\C^k}\big)^* S\big(\mathfrak{op}^A(\phi_n)\otimes\mathrm{Id}_{\C^k}\big)
    \end{equation*}
    is positive whenever $S\in\mathcal{B}_1(L^2(\R^d))\otimes\C^k$ is positive.

    As a last step in the proof we show that $\sum_{n=0}^\infty\phi_n\otimes\overline{\phi_n}=\lim_{N\rightarrow\infty}\Phi_N$ exists in some sense and is equal to the symbol of $\mathfrak{T}$. Since $\mathfrak{T}$ is bounded on $\mathcal{B}_1(L^2(\R^d))$, then $\mathfrak{T}$ is also bounded from $\mathcal{B}(\mathscr{S}'(\R^d),\mathscr{S}(\R^d))$ into $\mathcal{B}(\mathscr{S}(\R^d),\mathscr{S}'(\R^d))$. Hence $\mathfrak{T}=\mathfrak{Op}^A(\Phi)$ for some $\Phi\in\mathscr{S}(\R^{4d})$. Let $F\in\mathscr{S}(\R^{4d})$. Then using \cite[Lemma 2.2]{Thorn2026} we get:
    \begin{equation}\label{eq:cptpweak}
        \begin{aligned}
            \langle K^A\Phi,F\rangle_{\mathscr{S}',\mathscr{S}}&=\sum_{\tilde{\alpha},\tilde{\beta},\tilde{\gamma},\tilde{\delta}\in\Z^{2d}}\mathbb{M}_{(\tilde{\alpha},\tilde{\beta}),(\tilde{\gamma},\tilde{\delta})}^A[\mathfrak{Op}^A(\Phi)]\langle F,\overline{\mathcal{G}_{\tilde{\gamma}}^A}\otimes\mathcal{G}_{\tilde{\delta}}^A\otimes\mathcal{G}_{\tilde{\alpha}}^A\otimes\overline{\mathcal{G}_{\tilde{\beta}}^A}\rangle_{\mathscr{S}',\mathscr{S}}\\
            &=\sum_{\tilde{\alpha},\tilde{\beta},\tilde{\gamma},\tilde{\delta}\in\Z^{2d}}\lim_{N\rightarrow\infty}\mathbb{M}_{(\tilde{\alpha},\tilde{\beta}),(\tilde{\gamma},\tilde{\delta})}^A[\mathfrak{Op}^A(\Phi_N)]\langle F,\overline{\mathcal{G}_{\tilde{\gamma}}^A}\otimes\mathcal{G}_{\tilde{\delta}}^A\otimes\mathcal{G}_{\tilde{\alpha}}^A\otimes\overline{\mathcal{G}_{\tilde{\beta}}^A}\rangle_{\mathscr{S}',\mathscr{S}}
        \end{aligned}
    \end{equation}
    Since $\mathfrak{Op}^A(\Phi_N)$ converges strongly on $\mathcal{B}_1(L^2(\R^d))$ to $\mathfrak{Op}^A(\Phi)$ we can use the Banach-Steinhaus theorem \cite[Theorem 4.16]{Osborne2014} to conclude that
    \begin{align*}
        \sup_{N\in\N_0}\bigg|\mathbb{M}_{(\tilde{\alpha},\tilde{\beta}),(\tilde{\gamma},\tilde{\delta})}^A[\mathfrak{Op}^A(\Phi_N)]\bigg|&\leq C_1\sum_{\lambda\in\N_0^{4d},|\lambda|\leq k}\Big\Vert\langle\cdot\rangle^k\partial^\lambda\big(\mathcal{G}_{\tilde{\gamma}}^A\otimes\overline{\mathcal{G}_{\tilde{\delta}}^A}\otimes\overline{\mathcal{G}_{\tilde{\alpha}}^A}\otimes\mathcal{G}_{\tilde{\beta}}^A\big)\Big\Vert_{L^\infty}\\
        &\leq C_2\langle(\tilde{\alpha},\tilde{\beta},\tilde{\gamma},\tilde{\delta})\rangle^{l}
    \end{align*}
    for some $C_1,C_2,k,l>0$ independent of $\tilde{\alpha},\tilde{\beta},\tilde{\gamma},\tilde{\delta}$. \cite[Lemma 2.2]{Thorn2026} states that $\langle F,\overline{\mathcal{G}_{\tilde{\gamma}}^A}\otimes\mathcal{G}_{\tilde{\delta}}^A\otimes\mathcal{G}_{\tilde{\alpha}}^A\otimes\overline{\mathcal{G}_{\tilde{\beta}}^A}\rangle_{\mathscr{S}',\mathscr{S}}$ has super polynomial decay in $\tilde{\alpha},\tilde{\beta},\tilde{\gamma},\tilde{\delta}$, so we may use dominated convergence in \eqref{eq:cptpweak}, resulting in:
    \begin{align*}
        \langle K^A\Phi,F\rangle_{\mathscr{S}',\mathscr{S}}&=\lim_{N\rightarrow\infty}\sum_{\tilde{\alpha},\tilde{\beta},\tilde{\gamma},\tilde{\delta}\in\Z^{2d}}\mathbb{M}_{(\tilde{\alpha},\tilde{\beta}),(\tilde{\gamma},\tilde{\delta})}^A[\mathfrak{Op}^A(\Phi_N)]\langle F,\overline{\mathcal{G}_{\tilde{\gamma}}^A}\otimes\mathcal{G}_{\tilde{\delta}}^A\otimes\mathcal{G}_{\tilde{\alpha}}^A\otimes\overline{\mathcal{G}_{\tilde{\beta}}^A}\rangle_{\mathscr{S}',\mathscr{S}}\\
        &=\lim_{N\rightarrow\infty}\langle K^A\Phi_N,F\rangle_{\mathscr{S}',\mathscr{S}}
    \end{align*}
    So $K^A\Phi_N\rightarrow K^A\Phi$ as $N\rightarrow\infty$ in the weak*-topology. Finally, using the formal transpose\footnote{Recall that the formal transpose for an operator $T$ on Schwartz functions is an operator $S$ on Schwartz functions such that $\langle Tf,g\rangle_{\mathscr{S}',\mathscr{S}}=\langle f,Sg\rangle_{\mathscr{S}',\mathscr{S}}$.} of $K^A$ we conclude that $\Phi_N\rightarrow \Phi$ as $N\rightarrow\infty$ in the weak*-topology.
\end{proof}

\begin{remark}
    If $(\phi_n)_{n\in\N_0}$ are absolutely summable in $s_0(1)$, then $\sum_{n=0}^\infty\phi_n\otimes\overline{\phi_n}$ would converge absolutely in $s_0(1)\mathbin{\widehat{\otimes}}s_0(1)$ and so the limit would be a Hörmander super symbol since $s_0(1)\mathbin{\widehat{\otimes}}s_0(1)\subseteq S_0(1)$.
\end{remark}

\begin{remark}
    It may be quite hard to fulfill the condition:
    \begin{equation*}
        \sum_{n\in\N_0} \overline{\phi_n}\mathbin{\star^B}\phi_n=1
    \end{equation*}
    We mention shortly that if $B=0$ and we only let the $\phi_n$'s depend on the momentum variables, then $\overline{\phi_n}\mathbin{\star^0}\phi_n=|\phi_n|^2$, reducing the task to finding a sequence $(\phi_n)_{n\in\N_0}$ satisfying:
    \begin{equation*}
        \sum_{n=0}^\infty|\phi_n|^2=1
    \end{equation*}
    One example of such a sequence is $\phi_n$'s being a quadratic partition of unity.

    When $B\neq0$, one could perhaps change a sequence $(\phi_n)_{n\in\N_0}$ from the non-magnetic case to the magnetic case by considering the change of quantization $\phi_n^A=(k^A)^{-1}k^0\phi_n$. We will end this discussing for now and return to the problem of constructing suitable sequences of symbols in future work.
\end{remark}

We note that the assumption of pointwise convergence in Theorem \ref{thm:cptp} could be replaced by weak*-convergence in $\mathscr{S}'(\R^{2d})$ or convergence in $C^\infty(\R^{2d})$ with its canonical topology, see \cite[Proposition 1.1.2]{Nicola2010}.

%\section{Discussion}

%Sobolev spaces of tempered weights, extending several results

%inadequacies of frame techniques for proving ellipticity and positivity

%"Sobolev operator spaces", i.e. $\mathfrak{Op}^A(\Phi)\big(\mathcal{B}(L^2(\R^d))\big)$ where $\Phi$ is elliptic

%Lindbladians for the magnetic pseudo-differential super calculus, jump operators as quantizations. How this impacts the corresponding semi-group of cptp maps.

\begin{small}

\paragraph{Acknowledgements.}

This work was supported, in part, by the Danish National Research Foundation (DNRF), through the Center CLASSIQUE, grant nr. 187. HC also acknowledges support from Grant DFF 5281-00046B of the Independent Research Fund Denmark $|$ Natural Sciences.

\end{small}

\appendix

\section{Tempered Weights}\label{app:tempered_weights}

In this appendix we collect some results on tempered weights. As stated in Section \ref{sec:defn}, a tempered weight is a positive function $m\colon\R^n\rightarrow(0,\infty)$ such that there exists $a,C>0$ for which
\begin{equation*}
    m(x+y)\leq Cm(x)\langle y\rangle^a,\quad\forall x,y\in\R^n,
\end{equation*}
holds.

\begin{lemma}\label{lem:app_smoorth_tempered}
    For a tempered weight $m$ on $\R^n$ we may always find a smooth tempered weight $\tilde{m}$ such that there exists $C>0$ for which
    \begin{equation*}
        \frac{1}{C}m(x)\leq\tilde{m}(x)\leq Cm(x)
    \end{equation*}
    for all $x\in\R^n$.
\end{lemma}

\begin{proof}
    Choose a non-negative function $h\in C_c^\infty(\R^n)$ such that $\sum_{\alpha\in\Z^n}\tau_\alpha h\equiv1$ where $\tau_\alpha$ is a translation, i.e. $\big(\tau_\alpha f\big)(x)=f(x-\alpha)$. Then $\tilde{m}\colon\R^n\ni x\mapsto \sum_{\alpha\in\Z^n}m(\alpha)\big(\tau_\alpha h\big)(x)$ is a smooth tempered weight with behavior like $m$.
    
    It is indeed a tempered weight since for $x,y\in\R^n$ we have:
    \begin{align*}
        \tilde{m}(x+y)\tilde{m}(x)^{-1}&=\frac{\sum_{\alpha\in\Z^n}m(\alpha)\big(\tau_\alpha h\big)(x+y)}{\sum_{\beta\in\Z^n}m(\beta)\big(\tau_\beta h\big)(x)}\leq\frac{\sum_{\alpha\in\Z^n}C\langle\alpha\rangle^a\big(\tau_\alpha h\big)(x+y)}{\sum_{\beta\in\Z^n}C^{-1}\langle\beta\rangle^{-a}
        \big(\tau_\beta h\big)(x)}\\
        &\leq C_1\frac{\langle x+y\rangle^a\sum_{\alpha\in\Z^n}\langle x+y-\alpha\rangle^a\big(\tau_\alpha h\big)(x+y)}{\langle x\rangle^{-a}\sum_{\beta\in\Z^n}\langle x-\beta\rangle^{a}
        \big(\tau_\beta h\big)(x)}\leq C_2\langle y\rangle^a
    \end{align*}
    for constants $C_1,C_2>0$ and where $C,a>0$ satisfy \eqref{eq:peetre} with $m$.

    It also has behavior like $m$ since
    \begin{equation*}
        m(x)=\sum_{\alpha\in\Z^n}m(x)\big(\tau_\alpha h\big)(x)\leq C\sum_{\alpha\in\Z^n}\langle x-\alpha\rangle^am(\alpha)\big(\tau_\alpha h\big)(x)\leq C\Big(\sup_{\supp h}\langle\cdot\rangle^a\Big)\tilde{m}(x)
    \end{equation*}
    and
    \begin{equation*}
        \tilde{m}(x)\leq C\sum_{\alpha\in\Z^n}\langle\alpha-x\rangle^a m(x)\big(\tau_\alpha h\big)(x)\leq C\sup_{\supp h}\langle\cdot\rangle^a\sum_{\alpha\in\Z^n} m(x)\big(\tau_\alpha h\big)(x)=C\Big(\sup_{\supp h}\langle\cdot\rangle^a\Big)m(x)
    \end{equation*}
    holds for all $x\in\R^n$.
\end{proof}

\begin{lemma}\label{lem:app_lebesgue}
    For a tempered weight $m$ on $\R^n$ and $p\in(0,\infty]$, $m\in L^p(\R^n)$ if and only if $m$ is measurable and $(m(\alpha))_{\alpha\in\Gamma}\in\ell^p(\Gamma)$ for some lattice $\Gamma\coloneq L\Z^n$ with $L$ a non-singular $n\times n$ matrix.
\end{lemma}

\begin{proof}
    The proof is trivial when $p=\infty$. When $p<\infty$, choose a non-negative function $h\in C_c^\infty(\R^n)$ such that $\sum_{\alpha\in\Gamma}\tau_\alpha h\equiv1$. Then
    \begin{equation*}
        \int m(x)^p\mathrm{d}x=\sum_{\alpha\in\Gamma}\int m(x)^p\big(\tau_\alpha h\big)(x)\mathrm{d}x\leq C\sum_{\alpha\in\Gamma}m(\alpha)^p\int \big(\tau_\alpha(\langle \cdot\rangle^{ap}h)\big)(x)\mathrm{d}x\leq C_1\sum_{\alpha\in\Gamma}m(\alpha)^p
    \end{equation*}
    and
    \begin{equation*}
        \sum_{\alpha\in\Gamma}m(\alpha)^p= C_2\sum_{\alpha\in\Gamma}\int m(\alpha)^p\big(\tau_\alpha h\big)(x)\mathrm{d}x\leq C_3\sum_{\alpha\in\Gamma}\int m(x)^p\big(\tau_\alpha h\big)(x)\mathrm{d}x=C_3\int m(x)^p\mathrm{d}x
    \end{equation*}
    for constants $C_1,C_2,C_3>0$ and where $C,a>0$ satisfy \eqref{eq:peetre} with $m$. Note that we used monotone convergence to interchange the sums and integrals.
\end{proof}

\begin{lemma}\label{lem:app_tozero}
    For a tempered weight $m$ on $\R^n$, $m$ decays to zero at infinity if and only if some enumeration of  $(m(\alpha))_{\alpha\in\Gamma}$ converges to zero for some lattice $\Gamma\coloneq L\Z^n$ with $L$ a non-singular $n\times n$ matrix.
\end{lemma}

\begin{proof}
    Let $\iota\colon\N\ni k\mapsto\Gamma$ be an enumeration of $\Gamma$. If $m(x)\rightarrow0$ as $x\rightarrow\infty$, then for every $\varepsilon>0$ there exist $R>0$ such that $\sup_{|x|>R}m(x)<\varepsilon$. For large enough $N\in\N$ we have $|\iota(k)|>R$ for all $k\geq N$, so $m(\iota(k))<\varepsilon$ when $k\geq N$. Thus $m(\iota(k))\rightarrow0$ as $k\rightarrow\infty$.

    For the converse, let some enumeration $\iota$ of $\Gamma$ be given such that $m(\iota(k))\rightarrow0$ as $k\rightarrow\infty$. Let $\varepsilon>0$ be given and find $N\in\N$ such that $m(\iota(k))<\varepsilon$ when $k\geq N$. Define $R=1+\sup_{k<N}|\iota(\alpha)|$. Then for $|x|>R$ we have
    \begin{equation*}
        m(x)\leq C\inf_{k>N}m(\iota(k))\langle x-\iota(k)\rangle^a<C_1\varepsilon
    \end{equation*}
    for constants $C_1>0$ and where $C,a>0$ satisfy \eqref{eq:peetre} with $m$.
\end{proof}

\begin{small}
    
\end{small}

\end{document}